\documentclass[12pt,onecolumn]{IEEEtran}

\IEEEoverridecommandlockouts
\usepackage{amsmath,amssymb,amsthm,mathrsfs}
\usepackage{graphicx}
\usepackage{cite}
\usepackage{url}

\newcommand*{\QEDB}{\hfill\ensuremath{\square}}

\makeatletter%
\if@twocolumn%
\newcommand{\Figwidth}{\columnwidth}%
\def\twocolbreak{\nonumber\\ &}%
\else
\newcommand{\Figwidth}{3.5in}%
\def\twocolbreak{}%
\fi%
\makeatother%

\begin{document}

\title{Lattice Coding and Decoding for Multiple-Antenna Ergodic Fading Channels}
\author{\IEEEauthorblockN{Ahmed Hindy,~\IEEEmembership{Student member,~IEEE} and 
Aria Nosratinia,~\IEEEmembership{Fellow,~IEEE}}
\thanks{The authors are with the department of Electrical Engineering, University of Texas at Dallas, Email: ahmed.hindy@utdallas.edu and aria@utdallas.edu}
\thanks{This work was supported in part by the grant 1546969 from the National Science Foundation.}
}

\maketitle




\newtheorem{theorem}{Theorem}
\newtheorem{lemma}{Lemma}
\newtheorem{remark}{Remark}
\newtheorem{corollary}{Corollary}

\def\Px{P_x}
\def\Ph{\sigma_h^2}
\def\SNR{\rho}
\def\Ix{\boldsymbol{I}}
\def\hv{\boldsymbol{h}}
\def\xv{\boldsymbol{x}}
\def\yv{\boldsymbol{y}}
\def\wv{\boldsymbol{w}}
\def\ev{\boldsymbol{e}}
\def\tv{\boldsymbol{t}}
\def\dv{\boldsymbol{d}}
\def\zv{\boldsymbol{z}}
\def\vv{\boldsymbol{v}}
\def\uv{\boldsymbol{u}}
\def\av{\boldsymbol{a}}
\def\gv{\boldsymbol{g}}
\def\latpoint{\boldsymbol{\lambda}}
\def\SNRm{\boldsymbol{\Psi}}
\def\Hm{\boldsymbol{H}}
\def\Am{\boldsymbol{A}}
\def\Bm{\boldsymbol{B}}
\def\Um{\boldsymbol{U}}
\def\Dm{\boldsymbol{D}}
\def\Vm{\boldsymbol{V}}
\def\Gm{\boldsymbol{G}}
\def\Fm{\boldsymbol{F}}
\def\comp{\mathbb{C}}
\def\intp{\mathbb{Z}^+}
\def\real{\mathbb{R}}
\def\Hs{\mathcal{H}}
\def\Cs{\mathbb{L}}
\def\lat{\Lambda}
\def\voronoi{\mathcal{V}}
\def\genvoronoi{\Omega_1}
\def\codebook{\mathcal{L}_1}
\def\Ex{\mathbb{E}}
\def\vol{\text{Vol}}
\def\gap{\mathcal{G}}
\def\ball{\mathcal{B}}
\def\cov{\Sigma}
\def\covx{\boldsymbol{K_x}}
\def\Sset{\mathcal{S}}
\def\prob{\mathbb{P}}


\makeatletter%
\if@twocolumn%
\else
\vspace*{-0.4in}
\fi%
\makeatother%

\begin{abstract}

For ergodic fading, a lattice coding and decoding strategy is proposed and its performance is analyzed for the single-input single-output (SISO) and multiple-input multiple-output (MIMO) point-to-point channel as well as the multiple-access channel (MAC), with channel state information available only at the receiver (CSIR). At the decoder a novel strategy is proposed consisting of a time-varying equalization matrix followed by decision regions that depend only on channel statistics, not individual realizations. Our encoder has a similar structure to that of Erez and Zamir. For the SISO channel, the gap to capacity is bounded by a constant  under a wide range of fading distributions. For the MIMO channel under Rayleigh fading, the rate achieved is within a gap to capacity that does not depend on the signal-to-noise ratio (SNR), and diminishes with the number of receive antennas. The analysis is extended to the $K$-user MAC where similar results hold. Achieving a small gap to capacity while limiting the use of CSIR to the equalizer highlights the scope for efficient decoder implementations, since decision regions are fixed, i.e., independent of channel realizations. 
\end{abstract}

\begin{IEEEkeywords}
Ergodic capacity, ergodic fading, lattice codes, MIMO, multiple-access channel.
\end{IEEEkeywords}


\section{Introduction}
\label{intro}

In practical applications, structured codes are favored due to computational complexity issues; lattice codes are an important class of structured codes that has gained special interest in the last few decades.
An early attempt to characterize the performance of lattice codes in the additive white Gaussian noise (AWGN) channel was made by de~Buda~\cite{debuda_capacity}; a result that was later corrected by Linder \textit{et al.}~\cite{Linder}. Subsequently, Loeliger~\cite{Loeliger} showed the achievability of $\frac{1}{2} \log(\text{SNR})$ with lattice coding and decoding. Urbanke and Rimoldi~\cite{Urbanke_lattices} showed the achievability of $\frac{1}{2} \log(1+\text{SNR})$ with maximum-likelihood decoding.
Erez and Zamir~\cite{Erez_Zamir} demonstrated that lattice coding and decoding achieve the capacity of the AWGN channel using a method involving common randomness via a dither variable and minimum mean-square error (MMSE) scaling at the receiver. 
Subsequently, Erez {\em et al.}~\cite{Lattices_good} proved the existence of lattices with good properties that achieve the performance promised in~\cite{Erez_Zamir}. El~Gamal \textit{et al.}~\cite{DMT_lattices} showed that lattice codes achieve the capacity of the AWGN MIMO channel, as well as the optimal diversity-multiplexing tradeoff under quasi-static fading. Prasad and Varanasi~\cite{Varanasi1} developed lattice-based methods to approach the diversity of the MIMO channel with low complexity. Dayal and Varanasi~\cite{Varanasi2} developed diversity-optimal codes for Rayleigh fading channels using finite-constellation integer lattices and maximum-likelihood decoding. Zhan \textit{et al.}~\cite{IF_RX} introduced {\em integer-forcing linear receivers} as an efficient decoding approach that exploits the linearity of lattice codebooks. Ordentlich and Erez~\cite{IF_Erez} showed that in conjunction with precoding, integer-forcing can operate within a constant gap to the MIMO channel capacity. 
Going beyond the point-to-point channel, Song and Devroye~\cite{Relays_lattices_Devroye} investigated the performance of lattice codes in the Gaussian relay channel.
Nazer and Gastpar~\cite{CF_Nazer} introduced the compute-and-forward relaying strategy based on the decoding of integer combinations of interfering lattice codewords from multiple transmitters. Compute-and-forward was also an inspiration for the development of integer-forcing~\cite{IF_RX}.
\"{O}zg\"{u}r and Diggavi~\cite{Relays_lattices_Diggavi} showed that lattice codes can operate within a constant gap to the capacity of Gaussian relay networks. 
Ordentlich \textit{et al.}~\cite{CF_MAC} proposed lattice-based schemes that operate within a constant gap to the sum capacity of the~$K$-user MAC, and the sum capacity of a class of $K$-user symmetric Gaussian interference channels. 
On the other hand, a brief outline of related results on ergodic capacity is as follows. The ergodic capacity of the Gaussian fading channel was established by McEliece and Stark~\cite{stark_paper}. The capacity of the ergodic MIMO channel was established by Telatar~\cite{telatar} and Foschini and Gans~\cite{foschini}. 
The capacity region of the ergodic MIMO MAC was found by Shamai and Wyner~\cite{MAC_Shamai}. The interested reader is also referred to the surveys on fading channels by Biglieri \textit{et al.}~\cite{fading_Shamai} and Goldsmith \textit{et al.}~\cite{capacity_MIMO}.

For the most part, lattice coding results so far have addressed channel coefficients that are either constant or quasi-static. 
Vituri~\cite{dispersion} studied the performance of lattice codes with unbounded power constraint under regular fading channels. 
Recently, Luzzi and Vehkalahti~\cite{Luzzi_j} showed that a class of lattices belonging to a family of division algebra codes achieve rates within a constant gap to the ergodic capacity at all SNR, where the gap depends on the algebraic properties of the code as well as the antenna configuration. Unfortunately, the constant gap in~\cite{Luzzi_j} can be shown to be quite large at many useful antenna configurations, in addition to requiring substantial transmit power to guarantee any positive rate. Liu and Ling~\cite{polar_lattices} showed that polar lattices achieve the capacity of the i.i.d. SISO fading channel. Campello \textit{et al.}~\cite{Belfiore} also proved that algebraic lattices achieve the ergodic capacity of the SISO fading channel.

In this paper we propose a lattice coding and decoding strategy and analyze its performance for a variety of MIMO ergodic channels, showing that the gap to capacity is small at both high and low SNR.
The fading processes in this paper are finite-variance stationary and ergodic. First, we present a lattice coding scheme for the MIMO point-to-point channel under isotropic fading, whose main components include the class of nested lattice codes proposed in~\cite{Erez_Zamir} in conjunction with a time-varying MMSE matrix at the receiver. The proposed decision regions are spherical and depend only on the channel distribution, and hence the decision regions remain unchanged throughout subsequent codeword transmissions.%
\footnote{Although the decision regions are designed independently of the channel realizations, the received signal is multiplied by an MMSE matrix prior to decoding the signal, and hence channel knowledge at the receiver remains necessary for the results in this paper.} 
The relation of the proposed decoder with Euclidean lattice decoding is also discussed. The rates achieved are within a constant gap to the ergodic capacity for a broad class of fading distributions. Under Rayleigh fading, a bound on the gap to capacity is explicitly characterized which vanishes as the number of receive antennas grows. Similar results are also derived for the fading $K$-user MIMO MAC. 
The proposed scheme provides useful insights on the implementation of MIMO systems under ergodic fading. First, the results reveal that structured codes can achieve rates within a small gap to capacity. Moreover, channel-independent decision regions approach optimality when the number of receive antennas is large. Furthermore, for the special case of SISO channels the  gap to capacity is characterized for all SNR values and over a wide range of fading distributions. Unlike~\cite{Luzzi_j}, the proposed scheme achieves positive rates at low SNR where the gap to capacity vanishes. At moderate and high SNR, the gap to capacity is bounded by a constant that is independent of SNR and only depends on the fading distribution. In the SISO channel under Rayleigh fading, the  gap is a diminishing fraction of the capacity as the SNR increases.\footnote{Earlier versions of the SISO and MIMO point-to-point results of this paper appeared in~\protect\cite{paper1,paper3}; these results are improved in the current paper in addition to producing extensions to MIMO MAC.}

Throughout the paper we use the following notation. Boldface uppercase and lowercase letters denote matrices and column vectors, respectively. The set of real and complex numbers are denoted $\real, \comp$. $\boldsymbol{A}^T, \boldsymbol{A}^H$ denote the transpose and Hermitian transpose of matrix~$\boldsymbol{A}$, respectively. $a_i$ denotes element~$i$ of~$\boldsymbol{a}$. $\boldsymbol{A} \succeq \Bm$ indicates that $\boldsymbol{A}-\Bm$ is positive semi-definite. $\det(\Am)$ and $\text{tr}(\Am)$ denote the determinant and trace of $\Am$, respectively. $\prob,\Ex$ denote the probability and expectation operators, respectively. $\ball_n(q)$ is an $n$-dimensional sphere of radius~$q$ and the volume of an arbitrary shape $\mathcal{A}$ is $\vol(\mathcal{A})$.  All logarithms are in base~2. 


\section{Overview of Lattice Coding}
\label{sec:lattice}

A lattice $\lat$ is a discrete subgroup of $\real^n$ which is closed under reflection and real addition. 
The fundamental Voronoi region~$\voronoi$ of the lattice~$\lat$ is defined by
\begin{equation}
\voronoi = \big \{ \boldsymbol{s}: \text{arg} \min_{\latpoint \in \lat} || \boldsymbol{s} - \latpoint ||  = \boldsymbol{0} \big \}. 
\label{fund_voronoi}
\end{equation}
 
The {\em second moment per dimension} of  $\lat$ is defined as 
\begin{equation}
\sigma_{\lat}^2 = \frac{1}{n \vol(\voronoi)} \int_{\voronoi} ||\boldsymbol{s}||^2 d \boldsymbol{s},
\label{moment}
\end{equation}
and the {\em normalized second moment} $\mathit{G}(\lat)$ of $\lat$ is 
\begin{equation}
\mathit{G}(\lat) = \frac{\sigma_{\lat}^2}{\vol^{\frac{2}{n}}(\voronoi)},
\label{normalized_moment}
\end{equation}
where $\mathit{G}(\lat) > \frac{1}{2 \pi e}$ for any lattice in $\real^n$. 
Every $\boldsymbol{s} \in \real^n$ can be uniquely written as $\boldsymbol{s}= \latpoint + \ev$ where $\latpoint \in \lat$, $\ev \in \voronoi$. The quantizer is then defined by
\begin{equation}
Q_{\voronoi}(\boldsymbol{s})= \latpoint \;, \quad \text{if } \boldsymbol{s} \in \latpoint+\voronoi. 
\label{quantizer}
\end{equation}
Define the modulo-$\lat$ operation corresponding to $\voronoi$ as follows
\begin{equation}
[\boldsymbol{s}] \, \text{mod} \lat \triangleq \boldsymbol{s}- Q_{\voronoi}(\boldsymbol{s}).
\label{mod}
\end{equation}
The mod $\lat$ operation also satisfies
\begin{equation}
\big[ \boldsymbol{s} + \boldsymbol{t} \big] \, \text{mod} \lat = 
\big[ \boldsymbol{s} + [\boldsymbol{t}] \, \text{mod} \lat \big] \, \text{mod} \lat
\hspace{5mm} \forall \boldsymbol{s},\boldsymbol{t} \in \real^n.
\label{mod_mod}
\end{equation}
The lattice $\lat$ is said to be nested in~$\lat_1$ if~$\lat \subseteq \lat_1$. 
We employ the class of nested lattice codes proposed in \cite{Erez_Zamir}. The transmitter constructs a codebook $\codebook = \lat_1 \cap \voronoi$, whose rate
is given by
\begin{equation}
R=\frac{1}{n} \log \frac{\vol(\voronoi)}{\vol(\voronoi_1) } \,.
\label{lattice_rate}
\end{equation}

The coarse lattice~$\lat$ has an arbitrary second moment~$\Px$ and is good for covering and quantization, and the fine lattice~$\lat_1$ is good for AWGN coding, where both are construction-$A$ lattices~\cite{Loeliger,Erez_Zamir}. The existence of such lattices has been proven in~\cite{Lattices_good}.
A lattice $\lat$ is good for covering if 
\begin{equation}
\lim_{n \to \infty} \frac{1}{n} \log \frac{\vol(\ball_n(R_c))}{\vol(\ball_n(R_f))}=0,
\label{covering}
\end{equation}
where the covering radius $R_c$ is the radius of the smallest sphere spanning $\voronoi$ and ${R_f}$ is the radius of the sphere whose volume is equal to $\vol(\voronoi)$. In other words, for a good nested lattice code with second moment $\Px$, the Voronoi region~$\voronoi$ approaches a sphere of radius $\sqrt{n \Px}$. A lattice~$\lat$ is good for quantization if
\begin{equation}
\lim_{n \to \infty} \mathit{G}(\lat) = \frac{1}{2 \pi e} \, .
\label{quantization}
\end{equation}

A key ingredient of the lattice coding scheme proposed in \cite{Erez_Zamir} is using common randomness (dither)~$\dv$ in conjunction with the lattice code at the transmitter. $\dv$ is also known at the receiver, and is drawn uniformly over $\voronoi$.

\begin{lemma}\cite[Lemma~1]{Erez_Zamir}
If $\tv \in \voronoi$ is independent of $\dv$, then $\xv$ is uniformly distributed over $\voronoi$ 
and independent of the lattice point~$\tv$.
\label{lemma:uniform}
\end{lemma}

\begin{lemma}\cite[Theorem 1]{Lattice_quantization}.
An optimal lattice quantizer with second moment~$\sigma_{\lat}^2$ is white, and the autocorrelation of its dither~$\dv_{\text{opt}}$ is given by $\Ex [ \dv_{\text{opt}} \dv_{\text{opt}}^T ]= \sigma_{\lat}^2  \Ix_n$.
\label{lemma:x_distribution}
\end{lemma}

Note that the optimal lattice quantizer is a lattice quantizer with the minimum~$\mathit{G}(\lat)$. Since the proposed class of lattices is good for quantization, the autocorrelation of $\dv$ approaches that of~$\dv_{\text{opt}}$ as~$n$ increases. 
For a more comprehensive review on lattice codes see~\cite{lattice}.


\section{Point-to-point channel} 
\label{sec:ptp}

\subsection{MIMO channel}
\label{sec:ptp_MIMO}

Consider a MIMO point-to-point channel with $N_t$ transmit antennas and $N_r$ receive antennas. The received signal at time instant~$i$ is given by 
\begin{equation}
\yv_i=\Hm_i \xv_i+\wv_i,
\label{sig_Rx_MIMO}
\end{equation}
where $\Hm_i$ is an $N_r \times N_t$ matrix denoting the channel coefficients at time~$i$. The channel is zero-mean with strict-sense stationary and ergodic time-varying gain. Moreover,~$\Hm$ is isotropically distributed, i.e., $\prob (\Hm)= \prob (\Hm \Vm)$ for any unitary matrix $\Vm$ independent of $\Hm$. We first consider real-valued channels; the extension to complex-valued channels will appear later in this section.
 The receiver has instantaneous channel knowledge, whereas the transmitter only knows the channel distribution. $\xv_i \in \real^{N_t}$ is the transmitted vector at time~$i$, where the codeword
\begin{equation}
\xv \triangleq [\xv_1^T , \xv_2^T , \ldots , \xv_n^T]^T
\label{codeword_MIMO}
\end{equation} 
is transmitted throughout~$n$ channel uses and satisfies $\Ex [ || \xv ||^2 ] \leq n \Px$. 
The noise~$\wv \in \real^{N_r n}$ defined by $\wv^T \triangleq [\wv_1^T , \wv_2^T , \ldots , \wv_n^T]^T$ is a zero-mean i.i.d. Gaussian noise vector with covariance $ \Ix_{N_r n} $, and is independent of the channel realizations. For convenience, we define the SNR per transmit antenna to be $\SNR \triangleq \Px/ N_t$.


\begin{theorem}
For the ergodic fading MIMO channel with isotropic fading, any rate~$R$ satisfying 
\begin{equation}
R < -\frac{1}{2} \log \det \Big( \Ex \big[ (\Ix_{N_t} + \SNR \Hm^T \Hm)^{-1} \big] \Big) 
\label{rate_MIMO}
\end{equation}
is achievable using lattice coding and decoding.
\label{theorem:rate_MIMO}
\end{theorem}

\begin{proof}
{\em Encoding:} Nested lattice codes are used where $\lat \subseteq \lat_1$. The transmitter emits a lattice point~$\tv \in \lat_1$ that is dithered with~$\dv$ which is drawn uniformly over~$\voronoi$. $\lat$ has a second moment~$\Px$ and is good for covering and quantization, and $\lat_1$ is good for AWGN coding, where both are construction-$A$ lattices~\cite{Loeliger,Erez_Zamir}. The dithered codeword is then as follows
\begin{equation}
\xv = \, \big[\tv - \dv \big] \, \text{mod} \lat \, 
= \, \tv - \dv + \latpoint \, ,
\label{sig_tx}
\end{equation}
where $\latpoint= - Q_{\voronoi}(\tv - \dv ) \in \lat$ from \eqref{mod}. The coarse lattice $\lat \in \real^{N_t n}$ has a second moment~$\SNR$. The codeword is composed of $n$~vectors $\xv_i$ each of length~$N_t$ as shown in \eqref{codeword_MIMO}, which are transmitted throughout the $n$ channel uses.

{\em Decoding:} The received signal can be expressed in the form $\yv=\Hm_s \xv+\wv$, where  $\Hm_s$ is a block-diagonal matrix whose diagonal block~$i$ is $\Hm_i$. The received signal $\yv$ is multiplied by a matrix $\Um_s \in \real^{N_r n \times N_t n}$ and the dither is removed as follows
\begin{align}
\yv' \triangleq & \Um_s^T \yv + \dv \nonumber \\
 = & \xv + (\Um_s^T \Hm_s - \Ix_{N_t n}) \xv + \Um_s^T \wv + \dv \nonumber \\ 
 = & \tv + \latpoint + \zv,
\label{sig_rx_2}
\end{align}
where 
\begin{equation}
\zv \triangleq ( \Um_s^T \Hm_s - \Ix_{N_t n}) \xv + \Um_s^T \wv,
\label{eq_noise}
\end{equation}
and $\tv$ is independent of $\zv$, according to Lemma~\ref{lemma:uniform}. The matrix~$\Um_s$ that minimizes~$\Ex \big[ ||\zv||^2 \big]$ is then a block-diagonal matrix whose diagonal block~$i$ is the $N_t \times N_r$ MMSE matrix at time~$i$ given by
\begin{equation}
\Um_i= \SNR  (\Ix_{N_r} + \SNR \Hm_i \Hm_i^T)^{-1} \Hm_i.
\label{eq:U_MSE_MIMO}
\end{equation}

From \eqref{eq_noise},\eqref{eq:U_MSE_MIMO}, the equivalent noise at time~$i$, i.e., $\zv_i \in \real^{N_t}$, is expressed as
\begin{align}
\zv_i = & \Big ( \SNR  \Hm_i^T (  \Ix_{N_r} + \SNR \Hm_i \Hm_i^T  )^{-1} \Hm_i -  \Ix_{N_t}  \Big) \xv_i   \twocolbreak 
 + \SNR  \Hm_i^T ( \Ix_{N_r} + \SNR \Hm_i \Hm_i^T  )^{-1} \wv_i   \nonumber \\
= & - ( \Ix_{N_t} + \SNR \Hm_i^T \Hm_i  )^{-1} \xv_i + \SNR  \Hm_i^T ( \Ix_{N_r} + \SNR \Hm_i \Hm_i^T )^{-1} \wv_i,
\label{eq:zi_MIMO}
\end{align}
where~\eqref{eq:zi_MIMO} holds from the matrix inversion lemma, and  $\zv \triangleq [\zv_1^T, \ldots , \zv_n^T]^T$.
Naturally, the distribution of $\zv$ conditioned on $\Hm_i$ (which is known at the receiver) varies across time. For reasons that will become clear later, we need to get rid of this variation. Hence,  we ignore the instantaneous channel knowledge, i.e., the receiver considers $\Hm_i$ a random matrix after equalization.
The following lemma elaborates some geometric properties of~$\zv$ in the $N_t n$-dimensional space.

\begin{lemma}
\label{lemma:z_dist}
Let $\genvoronoi$ be a sphere defined by
\begin{equation}
\genvoronoi \triangleq \{ \vv \in \real^{N_t n} \, : \, || \vv||^2  \leq (1+\epsilon) \text{tr} ( \boldsymbol{\bar{\cov}} ) \}, 
\label{voronoi_MIMO} 
\end{equation} 
where $\boldsymbol{\bar{\cov}} \triangleq \SNR \, \Ex \big [ ( \Ix_{N_t n} + \SNR \Hm_s^T \Hm_s )^{-1} \big ]$. Then, for any $\epsilon>0$ and $\gamma>0$, there exists $n_{\gamma,\epsilon}$ such that for all $n>n_{\gamma,\epsilon}$, 
\begin{equation}
   \prob \big( \zv \notin \genvoronoi \big) < \gamma.
\label{eq:error_event}
\end{equation}
\end{lemma}

\begin{proof}
See Appendix~\ref{appendix:z_dist}.
\end{proof}

We apply a version of the ambiguity decoder proposed in~\cite{Loeliger} defined by the spherical decision region $\genvoronoi$ in~\eqref{voronoi_MIMO}.%
\footnote{$\genvoronoi$ satisfies the condition in~\cite{Loeliger} of being a bounded measurable region of~$\real^{N_t n}$, from~\eqref{voronoi_MIMO}.}
  The decoder chooses $\hat{\tv} \in \lat_1$ if the received point falls inside the decision region of the lattice point~$\hat{\tv}$, but not in the decision region of any other lattice point.

{\em Error Probability:} 
As shown in~\cite[Theorem~4]{Loeliger}, on averaging over the set of all good construction-A fine lattices~$\Cs$ of rate~$R$, the probability of error can be bounded by
\begin{align}
\frac{1}{|\Cs|} \sum_{\Cs_i \in \Cs} \, \prob_e  & <    \,  \prob (\zv \notin \genvoronoi) + (1+ \delta) \, \frac{\vol(\genvoronoi)}{\vol(\voronoi_1)} 
\twocolbreak 
 =  \, \prob (\zv \notin \genvoronoi) + (1+ \delta) 2^{nR} \, \frac{\vol(\genvoronoi)}{\vol(\voronoi)}, 
\label{error_prob_MIMO}
\end{align}
for any $\delta > 0$, where~\eqref{error_prob_MIMO} follows from~\eqref{lattice_rate}. This is a union bound involving
two events: the event that the noise vector is outside the decision region, i.e., $\zv \notin \genvoronoi$ and the event that the post-equalized point is in the intersection of two decision regions, i.e., $\big\{\yv' \in \{ \tv_1 + \genvoronoi \} \cap \{ \tv_2 + \genvoronoi\} \big\}$, where $\tv_1,~\tv_2 \in \lat_1$ are two distinct lattice points. Owing to Lemma~\ref{lemma:z_dist}, the probability of the first event vanishes with~$n$. Consequently, the error probability can be bounded by 
\begin{equation}
\frac{1}{|\Cs|} \sum_{\Cs_i \in \Cs} \, \prob_e < \gamma + (1+ \delta) 2^{nR} \frac{\vol(\genvoronoi)}{\vol(\voronoi)}, 
\label{error_prob_2_MIMO}
\end{equation}
for any $\gamma,\delta>0$. For convenience define~$\SNRm = \SNR \boldsymbol{\bar{\cov}}^{-1}$.
The volume of~$\genvoronoi$ is given by
\begin{equation}
\vol(\genvoronoi) = (1+ \epsilon)^{\frac{N_t n}{2}} 
 \vol \big( \ball_{N_t n} (\sqrt{N_t n \SNR}) \big) \,  \det \big ( \SNRm^{\frac{-1}{2}} \big) .
\label{vol_omega_MIMO}
\end{equation}

The second term in~\eqref{error_prob_2_MIMO} is bounded by
\begin{align}
& (1+ \delta) 2^{nR} (1+ \epsilon)^{N_t n/2} \frac{\vol(\ball_{N_t n} (\sqrt{N_t n \SNR}))}{\vol(\voronoi)} 
 \, \det \big ( \SNRm^{\frac{-1}{2}} \big ) \,  \nonumber \\
=&  \, (1+ \delta) 2^{ - N_t n \Big ( - \frac{1}{N_t n} \log \big( \frac{\vol(\ball_{N_t n} (\sqrt{N_t n \SNR}))}{\vol(\voronoi)} \big) + \xi \Big ) }, 
\label{eq:exponential_MIMO}
\end{align}
where 
\begin{align}
\xi \, \triangleq & \,  \frac{-1}{2} \log({1+ \epsilon}) - \frac{1}{2 N_t n} \log \det ( \SNRm^{-1} ) - \frac{1}{N_t} R   \nonumber \\
= & \, \frac{-1}{2} \log({1+ \epsilon}) - \frac{1}{2 N_t}  \log \det \big( \Ex \big[ (\Ix_{N_t} + \SNR \Hm^T \Hm)^{-1} \big] \big) 
\twocolbreak   - \frac{1}{N_t} R .
\label{xi_MIMO}
\end{align}
From~\eqref{covering}, since the lattice $\lat$ is good for covering, the first term of the exponent in~\eqref{eq:exponential_MIMO} vanishes. From~\eqref{eq:exponential_MIMO}, whenever $\xi$ is a positive constant we have $\prob_e  \to 0$ as $n \to \infty$, where $\xi$ is positive as long as
\begin{equation*}
R < \, -\frac{1}{2} \log \det \Big( \Ex \big[ (\Ix_{N_t} + \SNR \Hm^T \Hm)^{-1} \big] \Big)- \frac{1}{2} \log({1+ \epsilon}) - \epsilon',
\end{equation*}
where~$\epsilon,\epsilon'$ are positive numbers that can be made arbitrarily small by increasing~$n$.  From \eqref{sig_rx_2}, the outcome of the decoding process in the event of successful decoding is $\hat{\tv} = \tv + \latpoint$, where the transformation of $\tv$ by $\latpoint \in \lat$ does not involve any loss of information. Hence, on applying the modulo-$\lat$ operation on~$\hat{\tv}$%
\begin{equation}
[\hat{\tv}]\text{ mod}\lat \, = \, [\tv+\latpoint]\text{ mod}\lat \, = \, \tv,
\label{mod_lambda}
\end{equation}
where the second equality follows from~\eqref{mod_mod} since~$\latpoint \in \lat$. 
Since the probability of error in~\eqref{error_prob_2_MIMO} is averaged over the set of lattices in~$\Cs$, there exists at least one lattice that achieves the same (or less) error probability.%
\footnote{The error analysis adopted in this work (which stems from~\cite{Loeliger}) is based on {\em existence arguments} from the ensemble of construction-A lattices, i.e., the proof shows that at least one realization of the lattice ensemble achieves the average error performance. However, no guarantee that {\em all} members of the ensemble would perform similarly.}
 Following in the footsteps of~\cite{Erez_Zamir,DMT_lattices}, the existence of a sequence of covering-good coarse lattices with second moment~$\SNR$ that are nested in~$\lat_1$ can be shown.
The final step required to conclude the proof is extending the result to Euclidean lattice decoding, which is provided in the following lemma.
\begin{lemma}
The error probability of the Euclidean lattice decoder given by%
\footnote{The Euclidean decoder in~\eqref{ML_eqn} does not involve the channel realizations, unlike that in~\cite{DMT_lattices,Luzzi_j}.}
\begin{equation}
\hat{\tv}= \big [ \text{arg} \min_{\tv \in \lat_1} ||\yv'- \tv'||^2 \big ] \text{ mod} \lat 
\label{ML_eqn}
\end{equation} 
is upper-bounded by that of the ambiguity decoder in~\eqref{voronoi_MIMO}.
\label{lemma:lattice_decoder}
\end{lemma}
Details of the proof of Lemma~\ref{lemma:lattice_decoder} is provided in Appendix~\ref{appendix:decoder}, whose outline is as follows. For the cases where the ambiguity decoder declares a valid output ($\yv'$ lies exclusively within one decision sphere), both the Euclidean lattice decoder and the ambiguity decoder with spherical regions would be identical, since a sphere is defined by the Euclidean metric. However, for the cases where the ambiguity decoder {\em fails} to declare an output (ambiguity or atypical received sequence), the Euclidean lattice decoder still yields a valid output, and hence is guaranteed to achieve the same (or better) error performance, compared to the ambiguity decoder. This concludes the proof of Theorem~\ref{theorem:rate_MIMO}.
\end{proof}



The results can be extended to complex-valued channels with isotropic fading using a similar technique to that in~\cite[Theorem 6]{CF_Nazer}. The proof is omitted for brevity.

\begin{theorem}
For the ergodic fading MIMO channel with complex-valued channels $\tilde{\Hm}$ that are known at the receiver, any rate~$R$ satisfying 
\begin{equation}
R < \, - \log \det \Big( \Ex \big[ (\Ix_{N_t} + \SNR \tilde{\Hm}^H \tilde{\Hm})^{-1} \big] \Big) 
\label{rate_MIMO_complex}
\end{equation}
is achievable using lattice coding and decoding.  \QEDB
\label{theorem:rate_MIMO_complex}
\end{theorem}


We compare the achievable rate in \eqref{rate_MIMO_complex} with the ergodic capacity, given by~\cite{telatar}
\begin{equation}
C =   \Ex \big[ \log { \det ( \Ix_{N_t} + \SNR \tilde{\Hm}^H \tilde{\Hm} )} \big].
\label{capacity_MIMO_complex}
\end{equation}

\begin{corollary}
The gap $\gap$ between the rate of the lattice scheme \eqref{rate_MIMO_complex} and the ergodic capacity in \eqref{capacity_MIMO_complex} for the $N_t \times N_r$ ergodic fading MIMO channel is upper bounded by
\begin{itemize}
\item $N_r \geq N_t$ and $\SNR \geq 1$: For any channel for which all elements of $\Ex \big[ ( \tilde{\Hm}^H \tilde{\Hm}  )^{-1} \big]< \infty$ 
\begin{equation}
\gap < \log  \det \Big( \big( \Ix_{N_t} + \Ex [ \tilde{\Hm}^H \tilde{\Hm} ] ) \, \Ex \big[  ( \tilde{\Hm}^H \tilde{\Hm}  )^{-1}  \big] \Big). 
\label{gap_NtNr_eq}
\end{equation}
\item $N_r > N_t$ and $\SNR \geq 1$: When~$\tilde{\Hm}$ is i.i.d. complex Gaussian with zero mean and unit variance, 
\begin{equation}
\gap < N_t \, \log \big( 1+\frac{N_t+1}{N_r-N_t} \big).
\label{gap_NtNr_uneq}
\end{equation}
\item $N_t=1$ and and $\SNR < \frac{1}{\Ex [ || \tilde{\hv} ||^2 ]}$: When~$\Ex \big[ || \tilde{\hv} ||^4 \big] < \infty$, 
\begin{equation}
\gap < 1.45 \, \Ex \big[ ||\tilde{\hv}||^4 \big] \, \SNR^2.
\label{gap_SNRlow_SIMO}
\end{equation}
\end{itemize}
\label{cor:gap_MIMO}
\end{corollary}

\begin{proof}
See Appendix~\ref{appendix:gap_MIMO}. 
\end{proof} 

\begin{figure}
\centering
\includegraphics[width=\Figwidth]{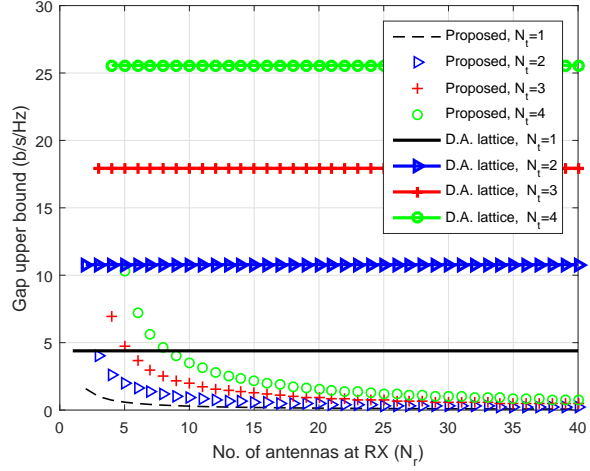}
\caption {\footnotesize{Guarantee on gap to capacity for Rayleigh fading MIMO valid for all $\SNR \geq 1$, shown for the proposed scheme as well as the division algebra lattices of~\cite{Luzzi_j} (denoted D.A. lattice).} 
\label{fig:gap_MIMO} } 
\end{figure}

The expression in~\eqref{gap_NtNr_uneq} for the Rayleigh fading case is depicted in Fig.~\ref{fig:gap_MIMO} for a number of antenna configurations. The gap-to-capacity vanishes with~$N_r$ for any~$\SNR \geq 1$.  This result has two crucial implications. First, under certain antenna configurations, lattice codes approximate the capacity at finite SNR. Moreover, channel-independent decision regions approach optimality for large~$N_r$. The results are also compared with that of the class of division algebra lattices proposed in~\cite{Luzzi_j}, whose gap-to-capacity is both larger and insensitive to~$N_r$. For the square MIMO channel with~$N_t=N_r=2$, the throughput of the proposed lattice scheme is plotted in Fig.~\ref{fig:rates_MIMO} and compared with that of~\cite{Luzzi_j}. The gap to capacity is also plotted, which show that for the proposed scheme the gap also saturates when $N_t=N_r$.

\begin{remark}
Division algebra codes in~\cite{Luzzi_j} guarantee non-zero rates only above a per-antenna SNR threshold that is no less than $21 N_t -1$ when $N_t < N_r$ and $\Ex[\tilde{\Hm}^H \tilde{\Hm}] = \Ix_{N_t}$, e.g., an SNR threshold of $10~dB$ for a $1\times 2$ channel. Our results guarantee positive rates at all SNR; for the single-input multiple-output (SIMO) channel at low SNR the proposed scheme has a gap on the order of $\SNR^2$. Since at $\SNR \ll 1$ we have $C \approx \SNR \Ex \big[ ||\tilde{h}||^2 \big] \log e$, the proposed scheme can be said to asymptotically achieve capacity at low SNR. Our results also show the gap diminishes to zero with large number of receive antennas under Rayleigh fading.
\end{remark}


\begin{figure}
\centering
\includegraphics[width=1\textwidth]{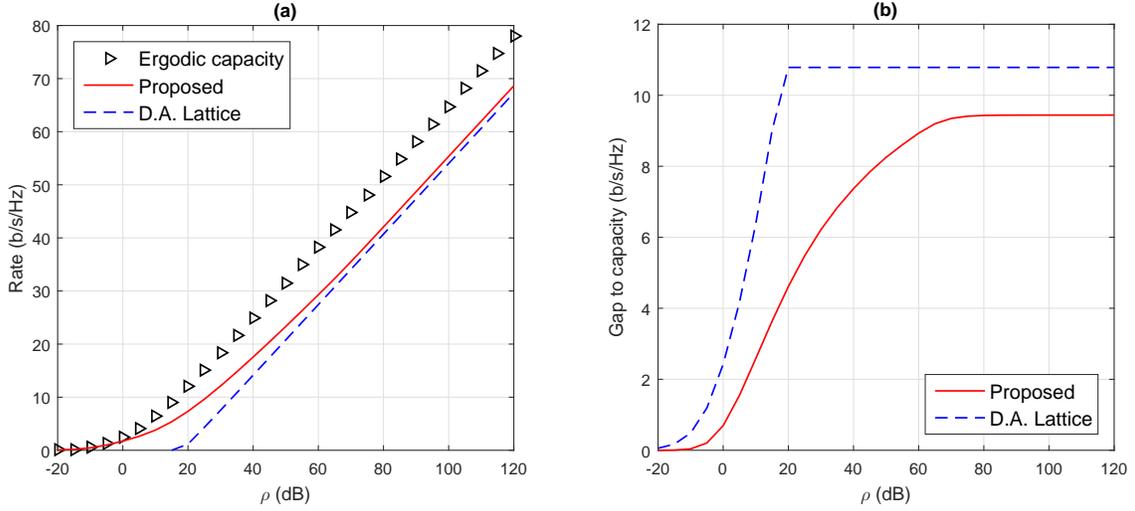}
\caption {\footnotesize{Rates achieved by the proposed lattice scheme vs. ergodic capacity under i.i.d. Rayleigh fading with $N_t=N_r=2$. } 
\label{fig:rates_MIMO} } 
\end{figure}

\subsection{SISO channel}
\label{sec:ptp_SISO}

For the case where each node is equipped with a single antenna, we find tighter bounds on the gap to capacity  for a wider range of fading distributions. Without loss of generality let $\Ex [|\tilde{h}|^2] = 1$.
The gap to capacity in the single-antenna case is given by
\begin{equation}
\gap = \Ex \big[ \log{( 1 + \SNR |\tilde{h}|^2 )} \big] + \log \big (\Ex [ \frac{1}{1+\SNR |\tilde{h}|^2 } ] \big ).
\label{gap_single}
\end{equation}

 In the following, we compute bounds on the gap for a  wide range of fading distributions, at both high and low SNR values.  

\begin{corollary}
When $N_t=N_r=1$, the gap to capacity $\gap$ is upper bounded as follows
\begin{itemize}
\item $\SNR <1$: For any fading distribution where $\Ex \big[ |\tilde{h}|^4 \big]  < \infty$,  
\begin{equation}
\gap < 1.45 \, \Ex \big[ |\tilde{h}|^4 \big] \, \SNR^2.
\label{gap_SNRlow}
\end{equation}
\item $\SNR \geq 1$: For any fading distribution where $\Ex \big[ \frac{1}{|\tilde{h}|^2} \big]  < \infty$, 
\begin{equation}
\gap < 1 + \log \Big( \Ex \big[ \frac{1}{|\tilde{h}|^2} \big] \Big).
\label{gap_SNRgen}
\end{equation}
\item $\SNR \geq 1$: Under Nakagami-$m$ fading with~$m>1$,
\begin{equation}
\gap < 1 + \log \big ( 1 + \frac{1}{m-1} \big ).
\label{gap_SNRnak}
\end{equation}
\item $\SNR \geq 1$: Under Rayleigh fading,
\begin{equation}
\gap < 0.48 + \log \big ( \log (1+ \SNR) \big ).
\label{gap_SNRgaus}
\end{equation}
\end{itemize}
\label{cor:gap_ptp}
\end{corollary}

\begin{proof}
See Appendix~\ref{appendix:gap_ptp}. 
\end{proof}

Although the gap depends on the SNR under Rayleigh fading, $\gap$ is a vanishing fraction of the capacity as $\SNR$ increases, i.e., $\lim_{\SNR \to \infty} \frac{\gap}{C}=0$. Simulations are provided to give a better view of Corollary~\ref{cor:gap_ptp}. First, the rate achieved under Nakagami-$m$ fading with $m=2$ and the corresponding gap to capacity are plotted in Fig.~\ref{fig:rates_Nakagami}. The performance is compared with that of the division algebra lattices from~\cite{Luzzi_j}. Similar results are also provided under Rayleigh fading in Fig.~\ref{fig:rates_Rayleigh}.

\begin{figure}
\centering
\includegraphics[width=1\textwidth]{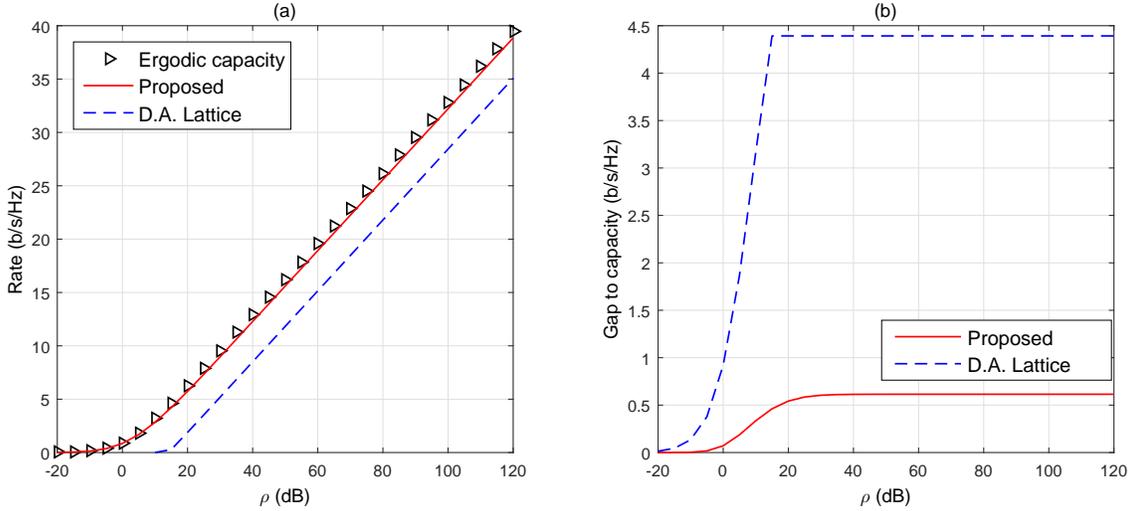}
\caption {\footnotesize{(a) The rates achieved by the lattice scheme vs. division algebra lattices~\cite{Luzzi_j} for SISO Nakagami-$m$ fading channels with $m=2$. (b)~Comparison of the gap to capacity.} 
\label{fig:rates_Nakagami} } 
\end{figure}

\begin{figure}
\centering
\includegraphics[width=1\textwidth]{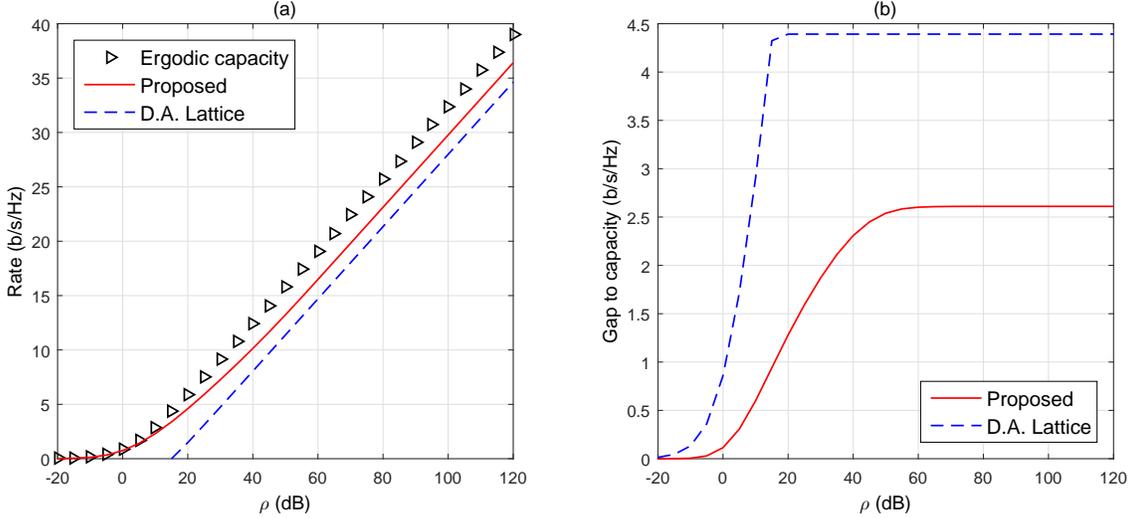}
\caption {\footnotesize{(a) The rates achieved by the lattice scheme vs. division algebra lattices~\cite{Luzzi_j} for SISO Rayleigh fading channels. (b)~
Comparison of the gap to capacity.} 
\label{fig:rates_Rayleigh} } 
\end{figure}


\begin{remark}
A closely related problem appears in~\cite{DMT_lattices}, where lattice coding and decoding were studied under quasi-static fading MIMO channels with CSIR, and a realization of the class of construction-$A$ lattices in conjunction with channel-matching decision regions (ellipsoidal shaped) were proposed. 
Unfortunately, this result by itself does not apply to ergodic fading because the application of the Minkowski-Hlawka Theorem~\cite[Theorem~1]{Loeliger}, on which the existence results of~\cite{DMT_lattices} depend, only guarantees the existence of a lattice for each channel realization, and is silent about the existence of a {\em universal} single lattice that is suitable for all channel realizations. This {\em universality} issue is the key challenge for showing results in the case of ergodic fading.\footnote{In~\cite{paper2} we attempted to show that for decoders employing {\em channel-matching decision regions} the gap to capacity vanishes, however, subsequently it was observed that~\cite{paper2} has not demonstrated the universality of the required codebooks.} 
The essence of the proposed lattice scheme in this section is approximating the ergodic fading channel (subsequent to MMSE equalization) with a non-fading additive-noise channel with lower SNR $\SNR' \triangleq \alpha \SNR$, where~$\alpha \leq 1$. The distribution of the (equivalent) additive noise term, $\zv$, in the approximate model depends on the fading distribution but not on the realization, which allows fixed decision regions for all fading realizations. The {\em SNR penalty factor}~$\alpha$ incurred from this approximation for the special case of $N_t=N_r=1$ is given by
\begin{equation}
\alpha = \Ex \big[ \frac{|\tilde{h}|^2}{\SNR |\tilde{h}|^2 + 1} \big]   \, \big  / \, \Ex [ \frac{1}{\SNR |\tilde{h}|^2 + 1} ].
\label{SNR_penalty} 
\end{equation}

As shown in the gap analysis throughout the paper, the loss caused by this approximation is small under most settings.
\end{remark}



\section{Multiple-Access Channel}
\label{sec:MAC}

\subsection{MIMO MAC}
\label{sec:MAC_MIMO}

Consider a $K$-user MIMO MAC with~$N_r$ receive antennas and~$N_{t_{k}}$ antennas at transmitter $k$. The received signal at time~$i$ is given by 
\begin{equation}
\tilde{\yv}^*_i= \tilde{\Hm}^*_{1,i} \tilde{\xv}^*_{1,i} + \tilde{\Hm}^*_{2,i} \tilde{\xv}^*_{2,i} + \ldots + \tilde{\Hm}^*_{K,i} \tilde{\xv}^*_{K,i} + \tilde{\wv}_i,
\label{sig_Rx_MIMO_MAC}
\end{equation}
where $\tilde{\Hm}_1^*,\ldots,\tilde{\Hm}_{K}^*$ are stationary and ergodic processes with zero-mean and complex-valued coefficients. The noise $\tilde{\wv}$ is circularly-symmetric complex Gaussian with zero mean and unit variance, and user~$k$ has a total power constraint~$N_{t_k} \SNR^*_k$. An achievable strategy for the $K$-user MIMO MAC is independent encoding for each antenna, i.e., user $k$ demultiplexes its data to $N_{t_k}$ data streams, and encodes each independently and transmits it through one of its antennas. The channel can then be analyzed as a SIMO MAC with $L \triangleq\sum_{k=1}^K N_{t_k}$ virtual users. The received signal is then given by
 \begin{equation}
\tilde{\yv}_i= \tilde{\hv}_{1,i} \tilde{x}_{1,i} + \tilde{\hv}_{2,i} \tilde{x}_{2,i} + \ldots + \tilde{\hv}_{L,i} \tilde{x}_{L,i} + \tilde{\wv}_i,
\label{sig_Rx_SIMO_MAC}
\end{equation}
where~$\tilde{\hv}_{\nu(k)+1,i}, \ldots, {\hv}_{\nu(k) + N_{t_k},i}$ denote the $N_{t_k}$ column vectors of~$\tilde{\Hm}^*_{k,i}$, and $\nu(k) \triangleq \sum_{j=1}^{k-1} N_{t_j}$. The virtual user~$\ell$ in~\eqref{sig_Rx_SIMO_MAC} has power constraint~$\SNR_l$, such that 
\begin{equation}
\SNR_{\nu(k)+1} + \ldots + \SNR_{\nu(k)+N_{t_k}} = N_{t_k} \SNR^*_k \, , \hspace{3mm}  k=1,2,\ldots,K.
\label{power_MAC}
\end{equation}

The MAC achievable scheme largely depends on the point-to-point lattice coding scheme proposed earlier, in conjunction with successive decoding. For the $L$-user SIMO MAC, there are~$L!$ distinct decoding orders, and the rate region is the {\em convex hull} of the $L!$ corner points. We define the one-to-one function~$\pi(\ell) \in \{ 1,2,\ldots,L \}$ that depicts a given decoding order. For example, $\pi(1)=2$ means that the codeword of user two is the first codeword to be decoded. 

\begin{theorem}
For the $L$-user SIMO MAC with ergodic fading and complex-valued channel coefficients, lattice coding and decoding achieve the following rate region
\begin{align}
 R_{MAC} \triangleq & \mathit{Co} \bigg( \bigcup_{\pi}  \Big \{ ( R_1,\ldots,R_L): 
\twocolbreak
R_{\pi(\ell)} \leq - \log \Big( \Ex \big[ \frac{1}{ 1 + \SNR_{\pi(\ell)} \tilde{\hv}_{\pi(\ell)}^H  \tilde{\Fm}_{\pi(\ell)}^{-1} \tilde{\hv}_{\pi(\ell)} } \big] \Big)  \Big \}  \bigg),
\label{rate_SIMO_MAC}
\end{align}
where
\begin{equation}
\tilde{\Fm}_{\pi(\ell)} \triangleq \Ix_{N_r} +  \sum_{j=\ell+1}^L \SNR_{\pi(j)} \tilde{\hv}_{\pi(j)} \tilde{\hv}_{\pi(j)}^H , 
\label{MAC_int}
\end{equation}
and $\mathit{Co}(\cdot)$ represents the {\em convex hull} of its argument, and the union is over all permutations $\big(\pi(1),\ldots,\pi(L) \big)$.
\label{theorem:rate_SIMO_MAC}
\end{theorem}

\begin{proof}
For ease of exposition we first assume the received signal is real-valued in the form $\yv_i = \sum_{\ell=1}^L \hv_{\ell,i} \xv_{\ell,i} +\wv_i$.

{\em Encoding:}
The transmitted lattice codewords are given by
\begin{equation}
\xv_{\ell}=[\tv_{\ell} - \dv_{\ell}] \, \text{mod} \lat^{(\ell)}  \hspace{5mm} \ell=1,2,\ldots,L,
\label{sig_tx_SIMO_MAC}
\end{equation}
where each lattice point $\tv_l$ is drawn from $\lat_1^{(\ell)} \supseteq \lat^{(\ell)}$, and the dithers $\dv_{\ell}$ are independent and uniform over~$\voronoi^{(\ell)}$. The second moment of~$\lat^{(\ell)}$ is~$\SNR_{\ell}$. Note that since transmitters have different rates and power constraints, each transmitter uses a different nested pair of lattices. The independence of the dithers across different users is necessary so as to guarantee the $L$~transmitted codewords are independent of each other.

{\em Decoding:} The receiver uses time-varying MMSE equalization and successive cancellation over $L$~stages, where in the first stage~$\xv_{\pi(1)}$ is decoded in the presence of $\xv_{\pi(2)}, \ldots, \xv_{\pi(L)}$ as noise, and then $\hv_{\pi(1),i} \xv_{\pi(1),i}$ is subtracted from~$\yv_i$ for $i=1, \ldots, n$. Generally, in stage~$\ell$, the receiver decodes~$\xv_{\pi(\ell)}$ from $\yv_{\pi(\ell)}$, where $\yv_{\pi(\ell),i} \triangleq \yv_i - \sum_{j=1}^{\ell-1} \hv_{\pi(j),i} \xv_{\pi(j),i}$. Note that at stage~$\ell$ the codewords $\xv_{\pi(1)}, \ldots, \xv_{\pi(\ell-1)}$ had been canceled-out in previous stages, whereas $\xv_{\pi(\ell+1)}, \ldots, \xv_{\pi(L)}$ are treated as noise. The MMSE vector at time~$i$, $\uv_{\pi(\ell),i}$, is given by
\begin{equation}
\uv_{\pi(\ell),i} = \SNR_{\pi(\ell)}  \big (  \Ix_{N_r} +  \sum_{j=\ell}^L \SNR_{\pi(j)} \hv_{\pi(j),i} \hv_{\pi(j),i}^T  \big ) ^{-1} \hv_{\pi(\ell),i} \, ,
\label{eq:ui_SIMO_MAC}
\end{equation}
and the equalized signal at time $i$ is expressed as follows
\begin{equation}
y'_{\pi(\ell),i}= \uv_{\pi(\ell),i}^T \yv_{\pi(\ell),i} + d_{\pi(\ell),i} = t_{\pi(\ell),i} + \lambda_{\pi(\ell),i} +  z_{\pi(\ell),i},
\label{sig_Rx_SIMO_MAC_eq}
\end{equation}
where $\latpoint_{\pi(\ell)} \in \lat^{(\pi(\ell))}$, and
\begin{align}
z_{\pi(\ell),i} = & \big( \uv_{\pi(\ell),i}^T \hv_{\pi(\ell),i} - 1  \big) x_{\pi(\ell),i} 
\twocolbreak
+  \sum_{j=\ell+1}^L  \uv_{\pi(\ell),i}^T \hv_{\pi(j),i} x_{\pi(j),i} + \uv_{\pi(\ell),i}^T \wv_i.
\label{z_eq_MAC}
\end{align}

Similar to the point-to-point step, we ignore the instantaneous channel state information subsequent to the MMSE equalization step. In order to decode~$\xv_{\pi(\ell)}$ at stage~$\ell$, we apply an ambiguity decoder defined by a spherical decision region 
\begin{align}
\genvoronoi^{(\pi(\ell))} \triangleq \Big\{ & \vv \in \real^{n}~:~ ||\vv||^2  \leq (1+\epsilon) n  \SNR_{\pi(\ell)} \,  \twocolbreak
 \Ex \Big[  \frac{1}
{1 + \SNR_{\pi(\ell)} \hv_{\pi(\ell)}^T \Fm_{\pi(\ell)}^{-1}  \hv_{\pi(\ell)} } \Big] \, \Ix_n \Big\}.
\label{voronoi_MAC}
\end{align}
where $\epsilon$ is an arbitrary positive constant.

{\em Error Probability:}
For an arbitrary decoding stage~$\ell$, the probability of error is bounded by
\begin{align}
\frac{1}{|\Cs|} \sum_{\Cs} \, \prob_e^{(\pi(\ell))}  <  \, & \prob (\zv_{\pi(\ell)} \notin \genvoronoi^{(\pi(\ell))}) \twocolbreak
+ (1+ \delta) 2^{n \check{R}_{\pi(\ell)}} \, \frac{\vol(\genvoronoi^{(\pi(\ell))})}{\vol(\voronoi^{(\pi(\ell))})},
\label{error_prob_MAC}
\end{align}
for some $\delta>0$. Following in the footsteps of the proof of Lemma~\ref{lemma:z_dist}, it can be shown that $\prob (\zv_{\pi(\ell)} \notin \genvoronoi^{(\pi(\ell))}) < \gamma$, where~$\gamma$ vanishes with~$n$; the proof is therefore omitted for brevity.
 From~\eqref{voronoi_MAC},
\begin{align}
\vol \big( \genvoronoi^{(\pi(\ell))} \big) = &(1+ \epsilon)^{\frac{n}{2}} 
 \vol \big( \ball_{n} (\sqrt{n \SNR_{\pi(\ell)}}) \big) \,   \twocolbreak 
\Big ( \Ex \Big[  \frac{1}
{1 + \SNR_{\pi(\ell)} \hv_{\pi(\ell)}^T \Fm_{\pi(\ell)}^{-1}  \hv_{\pi(\ell)} } \Big] \Big)^{\frac{n}{2}} .
\label{vol_omega_MAC}
\end{align}

The second term in~\eqref{error_prob_MAC} is then bounded by
\begin{equation}
 (1+ \delta) 2^{ -  n \Big ( - \frac{1}{n} \log \big( \frac{\vol(\ball_{n} (\sqrt{ n \SNR_{\pi(\ell)}}))}{\vol(\voronoi^{(\pi(\ell))})} \big) + \xi \Big ) }, 
\label{eq:exponential_MAC}
\end{equation}
where 
\begin{align}
\xi \, = & \,  -\frac{1}{2} \log \Big( \Ex \big[ \frac{1}{ 1 + \SNR_{\pi(\ell)} \hv_{\pi(\ell)}^T  \Fm_{\pi(\ell)}^{-1}  \hv_{\pi(\ell)} } \big] \Big)  -  \check{R}_{\pi(\ell)} \twocolbreak
- \frac{1}{2} \log({1+ \epsilon}). 
\label{xi_MAC}
\end{align}
The first term of the exponent in~\eqref{eq:exponential_MAC} vanishes since~$\lat^{(\pi(\ell))}$ is covering-good. Then, the error probability vanishes when
\begin{equation}
\check{R}_{\pi(\ell)} < -\frac{1}{2} \log \Big( \Ex \big[ \frac{1}{ 1 + \SNR_{\pi(\ell)} \hv_{\pi(\ell)}^T  \Fm_{\pi(\ell)}^{-1}  \hv_{\pi(\ell)} } \big] \Big) 
\label{rate_SIMO_MAC_1}
\end{equation}
for all $\ell \in \{1,2,\ldots,L \}$. The achievable rate region can then be extended to complex-valued channels, such that
\begin{equation}
R_{\pi(\ell)} < - \log \Big( \Ex \big[ \frac{1}{ 1 + \SNR_{\pi(\ell)} \tilde{\hv}_{\pi(\ell)}^H  \tilde{\Fm}_{\pi(\ell)}^{-1}  \tilde{\hv}_{\pi(\ell)} } \big] \Big) \, , \hspace{2mm} \ell=1,..,L.
\label{rate_SIMO_MAC_complex}
\end{equation}
This set of rates represents one corner point of the rate region. The whole rate region is characterized by the convex hull of the $L!$ corner points that represent all possible decoding orders, as shown in~\eqref{rate_SIMO_MAC}. This concludes the proof of Theorem~\ref{theorem:rate_SIMO_MAC}.
\end{proof}

Returning to the MIMO MAC model in~\eqref{sig_Rx_MIMO_MAC}, it is straightforward that the rate achieved by user~$k$ would then be 
\begin{equation}
R^*_k = \sum_{j=1}^{N_{t_k}} R_{\nu(k)+j}, 
\label{rate_MIMO_MAC_complex}
\end{equation}
where $R_j$ are the rates given in~\eqref{rate_SIMO_MAC_complex}. Now we compare $R_{\text{sum}} \triangleq \sum_{k=1}^K R^*_k$ with the sum capacity of the MIMO MAC model in~\eqref{sig_Rx_MIMO_MAC}. We focus our comparison on the case where the channel matrices have i.i.d. complex Gaussian entries and all users have the same number of transmit antennas as well as power budgets, i.e., $N_{t_k}=N_t,\, \SNR^*_k = \SNR$ for all $k \in \{1,2,\ldots,K\}$. The optimal input covariance is then a scaled identity matrix~\cite{telatar} and the sum capacity is given by~\cite{capacity_MIMO}
\begin{equation}
C_{\text{sum}} =  \Ex \Big[ \log { \det \big( \Ix_{N_r} + \sum_{k=1}^K \SNR \tilde{\Hm}_k^* \tilde{\Hm}_k^{*H} \big)} \Big].
\label{capacity_MIMO_MAC}
\end{equation}

\begin{corollary}
For the $K$-user fading MIMO MAC in~\eqref{sig_Rx_MIMO_MAC}, when $\tilde{\Hm}^*_k$ is i.i.d. complex Gaussian and~$N_r > K N_t$, the gap between the sum rate of the lattice scheme and the sum capacity at~$\SNR \geq 1$ is upper bounded by
\begin{equation}
\gap < \sum_{\ell=1}^{N_t K} \log \big( 1 + \frac{ \ell + 1 }{ N_r - \ell } \big).
\label{gap_MIMO_MAC_K}
\vspace{-5mm}
\end{equation}
\label{cor:gap_MIMO_MAC_K}
\end{corollary}

\begin{proof}
See Appendix~\ref{appendix:DoF_MIMO_MAC_K}. 
\end{proof}

Similar to the point-to-point MIMO, the gap to capacity vanishes at finite SNR as~$N_r$ grows, i.e., $\gap \to 0$ as $N_r \to \infty$. This suggests that decision regions which only depend on the channel statistics approach optimality for a fading MAC with large values of~$N_r$. 
The expression in~\eqref{gap_MIMO_MAC_K} is plotted in Fig.~\ref{fig:gap_MIMO_MAC} for $K=2$, as well as for the $K$-user SIMO MAC.

\begin{figure}
\centering
\includegraphics[width=1\textwidth]{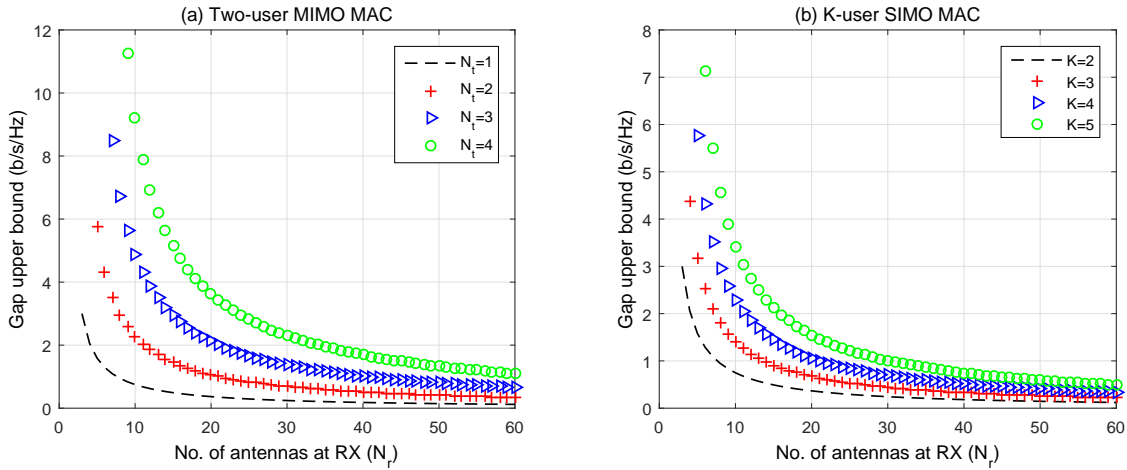}
\caption {\footnotesize{The upper bound on the gap to sum capacity of the MIMO MAC vs.~$N_r$.} 
\label{fig:gap_MIMO_MAC} } 
\end{figure}


\subsection{SISO MAC} 
\label{sec:MAC_SISO}

For the two-user case with $N_r = N_t =1$, the rate region in~\eqref{rate_SIMO_MAC} can be expressed by %
\footnote{Unlike the two-user MAC capacity region, the sum rate does not necessarily have a unit slope.}
\begin{align*}
R_1 < & - \gamma_1 \, ,  \\
R_2 < & - \gamma_2 \, ,
\end{align*}
\begin{equation}
(\gamma_4 - \gamma_2) R_1 + (\gamma_3 - \gamma_1) R_2 <    (\gamma_1 \gamma_2 - \gamma_3 \gamma_4),
\label{rate_MIMO_MAC_general}
\end{equation}
where
\begin{align*}
\gamma_1 = \log \big (\Ex \big[ \frac{1}{1 + \SNR_1 |\tilde{h}_1|^2 } \big] \big) \, ,  ~~~ & 
\gamma_2 = \log \big (\Ex \big[ \frac{1}{1 + \SNR_2 |\tilde{h}_2|^2 } \big] \big)  \, , \\
\gamma_3 = \log \big (\Ex \big[ \frac{1}{1 + \frac {\SNR_1 |\tilde{h}_1|^2}{1 + \SNR_2 |\tilde{h}_2|^2} } \big] \big )  \, , 
 & ~ \gamma_4 = \log \big (\Ex \big[ \frac{1}{1 + \frac {\SNR_2 |\tilde{h}_2|^2}{1 + \SNR_1 |\tilde{h}_1|^2} } \big] \big ) \, .
\end{align*}

For the case where all nodes are equipped with a single antenna, we characterize the gap to sum capacity of the two-user MAC for a wider range of distributions and over all SNR values. For ease of exposition we assume $\tilde{h}_1$ and $\tilde{h}_2$ are identically distributed with $\Ex [|\tilde{h}_1|^2] = \Ex [|\tilde{h}_2|^2] =1$. $\gap$ is then given by
\begin{align}
\gap \triangleq & \Ex \big[ \log{( 1 + \SNR |\tilde{h}_1|^2 + \SNR |\tilde{h}_2|^2)} \big] 
\twocolbreak
+ \log \Big( \Ex \big[ \frac{ 1 + \SNR |\tilde{h}_1|^2 }{ 1 + \SNR |\tilde{h}_1|^2 + \SNR |\tilde{h}_2|^2 } \big]
\Ex \big[ \frac{1}{1 + \SNR |\tilde{h}_1|^2} \big] \Big).
\label{gap_MAC_single}
\end{align}

\begin{corollary}
The gap to capacity of the two-user MAC given in~\eqref{gap_MAC_single} is upper-bounded as follows
\begin{itemize}
\item $\SNR < \frac{1}{2}$: For any fading distribution where $\Ex \big[ |\tilde{h}_1|^4 \big] < \infty$,
\begin{equation}
\gap < 1.45 \, \Big(1 + 2 \Ex \big[ |\tilde{h}_1|^4 \big] \Big) \, \SNR^2.
\label{gap_SNRlow_MAC}
\end{equation}
\item $\SNR \geq \frac{1}{2}$: For any fading distribution where $\Ex \big[ \frac{1}{|\tilde{h}_1|^2} \big]  < \infty$,
\begin{equation}
\gap < 2 + \log \Big( \Ex \big[ \frac{1}{|\tilde{h}_1|^2} \big] \Big).
\label{gap_SNRgen_MAC}
\end{equation}
\item $\SNR \geq \frac{1}{2}$: Under Nakagami-$m$ fading with~$m>1$,
\begin{equation}
\gap < 2 + \log \big ( 1 + \frac{1}{m-1} \big ).
\label{gap_SNRnak_MAC}
\end{equation}
\item $\SNR \geq \frac{1}{2}$: Under Rayleigh fading,
\begin{equation}
\gap < 1.48 + \log \big ( \log (1+ \SNR) \big ).
\label{gap_SNRgaus_MAC}
\end{equation}
\end{itemize}
\label{cor:gap_MAC}
\end{corollary}

\begin{proof}
See Appendix~\ref{appendix:gap_MAC}. 
\end{proof}

\begin{figure}
\centering
\includegraphics[width=1\textwidth]{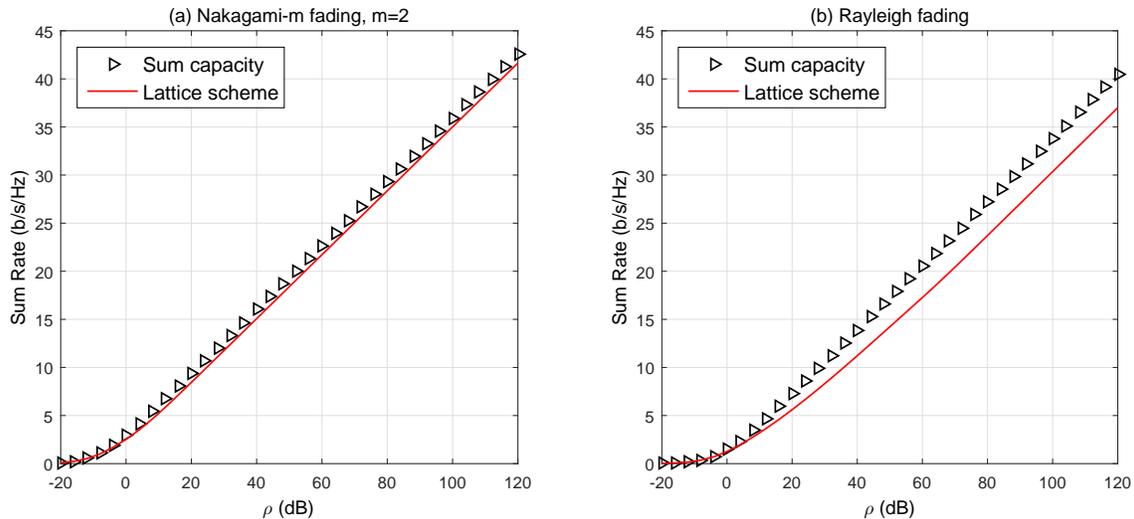}
\caption {\footnotesize{The two-user MAC sum rate vs. sum capacity.} 
\label{fig:rates_MAC} } 
\end{figure}

In Fig.~\ref{fig:rates_MAC}, the sum rate of the two-user MAC lattice scheme is compared with the sum capacity under Nakagami-$m$ fading with $m=2$, as well as under i.i.d. Rayleigh fading. It can be shown that the gap to sum capacity is small in both cases. Moreover, we plot the rate region under Rayleigh fading at $\SNR=-6~dB$ in Fig.~\ref{fig:MAC_region_low}. 
The rate region is shown to be close to the capacity region, indicating the efficient performance of the lattice scheme at low SNR as well.

\begin{figure}
\centering
\includegraphics[width=\Figwidth]{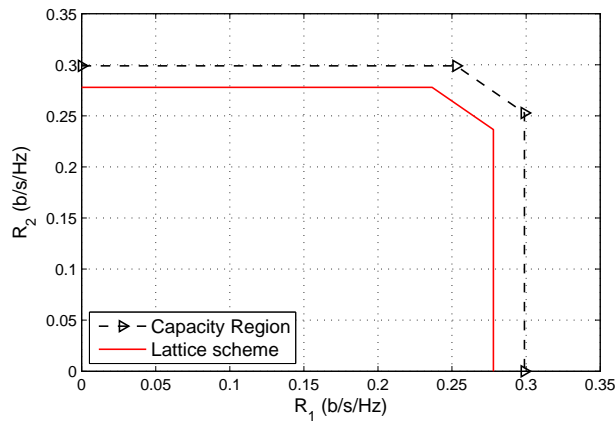}
\caption {\footnotesize{The two-user MAC rate region  vs. ergodic capacity region at $\SNR=-6~dB$ per user under i.i.d. Rayleigh fading.} 
\label{fig:MAC_region_low} } 
\end{figure}


\section{Conclusion} 
\label{sec:conc}

This paper presents a lattice coding and decoding strategy and analyzes its performance for a variety of ergodic fading channels.
For the MIMO point-to-point channel, the  rates achieved are within a constant gap to capacity for a large class of fading distributions. Under Rayleigh fading, the gap to capacity for the MIMO point-to-point and the $K$-user MIMO MAC vanishes as the number of receive antennas increases, even at finite SNR. The proposed decision regions are independent of the instantaneous channel realizations and only depend on the channel statistics. This both simplifies analysis and points to simplification in future decoder implementations. For the special case of single-antenna nodes, the gap to capacity is shown to be a constant for a wider range of fading distributions that include Nakagami fading. Moreover, at low SNR the gap to capacity is shown to be a diminishing fraction of the achievable rate. Similar results are also derived for the $K$-user MAC. Simulation results are provided that illuminate the performance of the proposed schemes.


\appendices

\section{Proof of Lemma~\ref{lemma:z_dist}}
\label{appendix:z_dist}

The aim of Lemma~\ref{lemma:z_dist} is showing that~$\zv$ lies with high probability within the sphere~$\genvoronoi$. However, computing the distribution of~$\zv$ is challenging since it depends on that of~$\xv, \Hm$ as shown in~\eqref{eq_noise}, where the distribution of~$\xv$ is not known at arbitrary block length, and no fixed distribution is imposed for~$\Hm$. The outline of the proof is as follows. First,  we replace the original noise sequence with a {\em noisier} sequence whose statistics are known. Then, we use the weak law of large numbers to show that the noisier sequence is confined with high probability within~$\genvoronoi$, which implies that the original noise~$\zv$ is also confined within~$\genvoronoi$. 
We decompose the noise $\zv$ in~\eqref{eq:zi_MIMO} in the form $\zv = \Am_s \xv + \sqrt{\SNR} \Bm_s \wv$, where both $\Am_s$, $\Bm_s$ are block-diagonal matrices with diagonal blocks $\Am_i \triangleq - ( \Ix_{N_t} + \SNR \Hm_i^T \Hm_i )^{-1}$ and $\Bm_i \triangleq \sqrt{\SNR}  \Hm_i^T ( \Ix_{N_r} + \SNR \Hm_i \Hm_i^T )^{-1}$, respectively, such that 
\begin{equation}
\Am_s \Am_s^T + \Bm_s \Bm_s^T = ( \Ix_{N_t n} + \SNR \Hm_s^T \Hm_s )^{-1}. 
\label{cov_sum}
\end{equation}
Since $\Hm_i$ is a stationary and ergodic process, $\Am_i$ and~$\Bm_i$ are also stationary and ergodic. Denote the eigenvalues of the random matrix $\Hm^T \Hm$ (arranged in ascending order) by $\sigma_{H,1}^2, \ldots, \sigma_{H,N_t}^2$. Then its eigenvalue decomposition is $\Hm^T \Hm \triangleq \Vm \Dm \Vm^T$, where $\Vm$ is a unitary matrix and $\Dm$ is a diagonal matrix whose {\em unordered} entries are $\sigma_{H,1}^2, \ldots, \sigma_{H,N_t}^2$. Owing to the isotropy of the distribution of~$\Hm$, $\Am \Am^T = \Vm (\Ix_{N_t}+ \SNR \Dm)^{-2} \Vm^T$ is {\em unitarily invariant}, i.e., $ \prob (\Am \Am^T) = \prob (\check{\Vm} \Am \Am^T \check{\Vm}^T)$ for any unitary matrix~$\check{\Vm}$ independent of~$\Am$. As a result $\Vm$ is independent of~$\Dm$~\cite{random_matrix}. Hence,
\begin{align}
\Ex \big[ \Am \Am^T \big] & = \,  \Ex \big[ ( \Ix_{N_t} + \SNR \Hm^T \Hm  )^{-2}  \big]  \,
\twocolbreak
= \, \Ex \big[ \Vm (\Ix_{N_t} + \SNR \Dm)^{-2} \Vm^T  \big]  \nonumber \\
& =  \Ex_{\Vm|\Dm} \big[ \Vm  \Ex_{\Dm} [ (\Ix_{N_t} + \SNR \Dm)^{-2} ] \Vm^T  \big]  \,
\twocolbreak
= \, \Ex_{\Vm|\Dm} \big[ \Vm  \sigma_A^2 \Ix_{N_t}  \Vm^T  \big]  \,
= \,  \sigma_A^2 \Ix_{N_t},
\label{A_cov_2}
\end{align}
where $\sigma_A^2 \triangleq \Ex_j \big[ \Ex_{\sigma_{\Hm,j}} [ \frac {1}{(1+ \SNR \sigma_{\Hm,j}^2)^2} ]   \big]$. Similarly, it can be shown that $\Ex \big[ \Bm \Bm^T \big] =  \sigma_B^2 \Ix_{N_t}$, where 
\[
\sigma_B^2 \triangleq \Ex_j \bigg[ \Ex_{\sigma_{\Hm,j}} \big[ \frac {\SNR \sigma_{\Hm,j}^2}{(1+ \SNR \sigma_{\Hm,j}^2)^2} \big]   \bigg].
\]
 For convenience define $\sigma_z^2 \triangleq \sigma_A^2+\sigma_B^2$. Next, we  compute the autocorrelation of $\zv$ as follows
\begin{equation}
\boldsymbol{\cov_z} \triangleq  \Ex \big[  \zv \zv^T  \big] = \Ex \big[ \Am_s \boldsymbol{\cov_x} \Am_s^T \big] + \SNR \Ex \big[ \Bm_s  \Bm_s^T \big],
\label{cov_1}
\end{equation}	
where $\boldsymbol{\cov_x} \triangleq \Ex \big[ \xv \xv^T \big]$. Unfortunately, $\boldsymbol{\cov_x}$ is not known for all~$n$, yet it approaches $\SNR \Ix_{N_t n}$ for large $n$, according to Lemma~\ref{lemma:x_distribution}. Hence one can rewrite 
\begin{align}
\boldsymbol{\cov_z} = & \, \underbrace { \sigma_x^2 \Ex \big[ \Am_s \Am_s^T  +  \Bm_s  \Bm_s^T \big] }_{\sigma_x^2 \, \sigma_z^2 \Ix_{N_t n}}  +
\twocolbreak
  \underbrace { \Ex \big[ \Am_s (\boldsymbol{\cov_x} - \sigma_x^2 \Ix_{N_t n}) \Am_s^T \big] +  (\SNR- \sigma_x^2) \Ex \big[ \Bm_s  \Bm_s^T \big] }_{ \succ 0 }  \,  ,
\label{cov_2}
\end{align}	
where 
$\sigma_x^2 \triangleq \lambda_{\text{min}}(\boldsymbol{\cov_x}) - \delta$, and $\lambda_{\text{min}}(\boldsymbol{\cov_x})$ is the minimum eigenvalue of $\boldsymbol{\cov_x}$. Note that $\SNR \geq \sigma_x^2$, from the definition in~\eqref{moment}. As a result the second term in~\eqref{cov_2} is positive-definite, and $\boldsymbol{\cov_z} \succ \sigma_x^2 \, \sigma_z^2 \Ix_{N_t n}$. This implies that
\begin{equation}
\boldsymbol{\cov_z}^{-1} \prec    \frac{1}{\sigma_x^2 \sigma_z^2} \Ix_{N_t n} \,.
\label{cov_3}
\end{equation}

To make noise calculations more tractable, we introduce a related noise variable that modifies the second term of~$\zv$ as follows
\begin{equation}
\zv^* = \Am_s \xv + \Bm_s \big(\sqrt{\SNR}  \wv + \sqrt{\frac{1}{N_t n} R_c^2 - \SNR} \, \wv^* \big),
\label{z_MIMO_2}
\end{equation}
where $\wv^*$ is i.i.d. Gaussian with zero mean and unit variance, and $R_c$ is the covering radius of~$\voronoi$. We now wish to bound the probability that $\zv^*$ is outside a sphere of radius $\sqrt{(1+\epsilon) N_t n \sigma_x^2 \sigma_z^2}$.
First, we rewrite 
\begin{align}
||\zv^*||^2 = & \xv^T \Am_s^T \Am_s \xv + \frac{1}{N_t n} R_c^2 \wv^T \Bm_s \Bm_s^T \wv 
\twocolbreak + 2 \sqrt{\frac{1}{N_t n} R_c^2} \, \xv^T \Am_s^T \Bm_s \wv.
\label{z_norm}
\end{align}
Then, we bound each term separately using the weak law of large numbers.  The third term satisfies %
\footnote{The third term in~\eqref{z_norm} is a sum of zero mean uncorrelated random variables to which the law of large numbers applies~\cite{LLN}.} 
\begin{equation}
\prob \big( 2 \sqrt{\frac{1}{N_t n} R_c^2} \, \xv^T \Am_s^T \Bm_s \wv > N_t n \epsilon_3 \big) < \gamma_3.
\label{z_norm_3}
\end{equation}
Addressing the second term in~\eqref{z_norm}, %
\footnote{Note that $\mu_i \triangleq \wv_i^T \Bm_i^T \Bm_i \wv_i$ is also a stationary and ergodic process that obeys the law of large numbers.}
\begin{align}
& \prob \big( \frac{1}{N_t n} R_c^2 \, \wv^T \Bm_s^T \Bm_s \wv > \sigma_B^2 R_c^2 + N_t n \epsilon_2 \big) 
\twocolbreak
= \prob \big( \frac{1}{N_t n} R_c^2 \, \text{tr} ( \Bm_s \wv \wv^T \Bm_s^T \wv ) > \sigma_B^2 R_c^2 + N_t n \epsilon_2 \big) <\gamma_2.
\label{z_norm_2} 
\end{align}
 Now, we bound the first term in~\eqref{z_norm}. Given that $\Am_s^T \Am_s$ is a block-diagonal matrix with $\Ex \big[ \Am_s^T \Am_s \big] = \sigma_A^2 \Ix_{N_t n}$, and that $\boldsymbol{\cov_x} \to \SNR \Ix_{N_t n}$ as $n \to \infty$, it can be shown using~\cite[Theorem 1]{weighted_avg} that $\frac{1}{||\xv||^2} \xv^T \Am_s^T \Am_s \xv \to \sigma_A^2$ as $n \to \infty$. More precisely,
\begin{align}
& \prob \big( \xv^T \Am_s^T \Am_s \xv   >  \sigma_A^2 R_c^2 + N_t n \epsilon_1 \big) 
\twocolbreak
< \prob \big( \xv^T \Am_s^T \Am_s \xv  >  \sigma_A^2 ||\xv||^2 + N_t n \epsilon_1 \big) < \gamma_1,
\label{z_norm_1}
\end{align}
where $||\xv||^2 < R_c^2$, and $\epsilon_1, \epsilon_2,\epsilon_3$ and $\gamma_1,\gamma_2,\gamma_3$ can be made arbitrarily small by increasing~$n$. Using a union bound, 
\begin{equation}
\prob \big( ||\zv^*||^2 > (1+\epsilon_4)  R_c^2 \sigma_z^2  \big) < \gamma,  
\label{z_norm_final_1}
\end{equation}
where $\epsilon_4 \triangleq \frac{(\epsilon_1+\epsilon_2+\epsilon_3)}{R_c^2 \sigma_z^2}$ and $\gamma \triangleq \gamma_1 + \gamma_2 + \gamma_3$. For large~$n$, $\frac{1}{N_t n}R_c^2 \leq (1+\epsilon_6) \SNR$ for covering-good lattices and $\SNR \leq (1+\epsilon_7) \sigma_x^2$ according to Lemma~\ref{lemma:x_distribution}. Let $\epsilon_5 \triangleq (1+\epsilon_6)(1+\epsilon_7)-1$, then for any $\epsilon$ such that $\epsilon \leq (1+\epsilon_4)(1+\epsilon_5) -1$,
\begin{align}
& \prob \big( \zv^{*T} \boldsymbol{\cov_z} \zv^* > (1+\epsilon)  N_t n   \big) 
\twocolbreak
< \prob \big( ||\zv^*||^2 > (1+\epsilon)  N_t n \SNR \sigma_z^2  \big)  
\label{z_norm_final_2} \\
& =   \prob \Big( \zv^{*T}  \big( \Ex \big[ (\Ix_{N_t n} + \SNR \Hm_s^T \Hm_s)^{-1}  \big] \big)^{-1} \zv^*  > (1+\epsilon)  N_t n \SNR   \Big)
\twocolbreak
 < \gamma,
\label{z_norm_final_3}
\end{align}
where~\eqref{z_norm_final_2} holds from~\eqref{cov_3} and~\eqref{z_norm_final_3} holds since $\Ex \big[ (\Ix_{N_t n} + \SNR \Hm_s^T \Hm_s)^{-1}  \big] = \sigma_z^2 \Ix_{N_t n}$, according to~\eqref{cov_sum}. The final step is to show that $||\zv^* - \zv|| \to 0$ as $n \to \infty$, where $\zv^*-\zv= \sqrt{\frac{1}{N_t n} R_c^2 - \SNR} \, \Bm_s \wv^*$. From the structure of~$\Bm_s$, the norm of each of its rows is less than~$N_t$, and hence the variance of each of the elements of~$\Bm_s \wv^*$ is no more than~$N_t$. Since $\lim_{n \to \infty} \frac{1}{N_t n} R_c^2 = \SNR$ for a covering-good lattice, it can be shown using Chebyshev's inequality that the elements of $\sqrt{\frac{1}{N_t n} R_c^2 - \SNR} \, \Bm_s \wv^*$ vanish and $|\zv_j^* - \zv_j| \to 0$ as $n \to \infty$ for all~$j \in \{ 1, \ldots, N_t n \}$. This concludes the proof of Lemma~\ref{lemma:z_dist}.


\section{Proof of Lemma~\ref{lemma:lattice_decoder}}
\label{appendix:decoder}

Denote by~$\Sset$ the event that the post-processed received point $\yv'$ falls exclusively within one decision sphere, defined in~\eqref{voronoi_MIMO}, where the probability of occurrence of~$\Sset$ is $\prob_{\Sset} \triangleq 1-\gamma_s$. Using the law of total probability, the probability of error (in general) is given by
\begin{equation}
\prob_e = \prob_{e|\Sset} \prob_{\Sset} + \prob_{e|\Sset^c} \prob_{\Sset^c}
\label{error_detailed}
\end{equation}
First we analyze the ambiguity decoder with spherical decision regions (denoted by superscript~$^{(SD)}$). From the definition of ambiguity decoding, $\prob_{e|\Sset^c}^{(SD)}=1$. Hence,
\begin{equation}
\prob_e^{(SD)} = \eta' (1-\gamma_s) + \gamma_s,
\label{error_detailed_spherical}
\end{equation}
where $\prob_{e|\Sset}^{(SD)} \triangleq \eta'$. Now we analyze the Euclidean lattice decoder (denoted by superscript~$ ^{(LD)}$). Since a sphere is defined by the Euclidean metric, the outcomes of the spherical decoder and the Euclidean lattice decoder {\em conditioned on the event~$\Sset$} are identical, and hence yield the same error probability, i.e., $\prob_{e|\Sset}^{(LD)} = \prob_{e|\Sset}^{(SD)} = \eta'$. However, from~\eqref{ML_eqn}, the Euclidean lattice decoder declares a valid output even under the event $\Sset^c$. Hence,  $\prob_{e|\Sset^c}^{(LD)} \triangleq \eta''\leq 1$. Thereby,
\begin{equation}
\prob_e^{(LD)} = \eta' (1-\gamma_s) + \eta'' \gamma_s \leq \prob_e^{(SD)}.
\label{error_detailed_lattice}
\end{equation}


\section{Proof of Corollary~\ref{cor:gap_MIMO} }
\label{appendix:gap_MIMO}

\begin{lemma} 
For an i.i.d. complex Gaussian~$M \times N$ matrix~$\Gm$ whose elements have zero mean, unit variance and $M > N$, then $\Ex \big[ (\Gm^H \Gm)^{-1} \big] = \frac{1}{M-N} \, \Ix_N$.
\label{lemma:wishart}
\end{lemma}

\begin{proof}
See~\cite[Section~V]{Wishart}.
\end{proof}

\subsection{Case $1$: $N_r \geq N_t$ and the elements of~$\Ex \big[ (\tilde{\Hm}^H \tilde{\Hm})^{-1} \big] < \infty$}

\begin{align*}
\gap = & \, C-R   \\
= &  \Ex \big[ \log \det ( \Ix_{N_t} + \SNR \tilde{\Hm}^H \tilde{\Hm} ) \big] 
\twocolbreak
+  \log \det \Big( \Ex \big[ (\Ix_{N_t} + \SNR \tilde{\Hm}^H \tilde{\Hm})^{-1} \big] \Big) \\
\stackrel{(a)}{\leq} & \log \det \Big( \Ix_{N_t} + \SNR \Ex [ \tilde{\Hm}^H \tilde{\Hm} ] \Big)  
\twocolbreak
+  \log \det \Big( \Ex \big[ (\Ix_{N_t} + \SNR \tilde{\Hm}^H \tilde{\Hm})^{-1} \big] \Big) \\
\stackrel{(b)}{<} &  \log \det \Big( \Ix_{N_t} + \SNR \Ex [ \tilde{\Hm}^H \tilde{\Hm} ] \Big) 
+  \log \det \Big( \Ex \big[ ( \SNR \tilde{\Hm}^H \tilde{\Hm})^{-1} \big] \Big) \\
= & \log \det \Big(  \big( \frac{1}{\SNR} \Ix_{N_t} + \Ex [ \tilde{\Hm}^H \tilde{\Hm} ] \big)  \Ex \big [ (\tilde{\Hm}^H \tilde{\Hm})^{-1} \big ] \Big) \\ 
\stackrel{(c)}{\leq} & \log \det \Big(  \big( \Ix_{N_t} + \Ex [ \tilde{\Hm}^H \tilde{\Hm} ] \big)  \Ex \big [ (\tilde{\Hm}^H \tilde{\Hm})^{-1} \big ] \Big),
\end{align*}
where~$(a), \, (b)$ follow since $\log \det (\Am)$ is a concave and non-decreasing function over the set of all positive definite matrices~\cite{convex}. $(c)$ follows since $\SNR \geq 1$.

\subsection{Case $2$: $N_r > N_t$ and the elements of~$\tilde{\Hm}$ are i.i.d. complex Gaussian with zero mean and unit variance}

\begin{align*}
\gap  & \stackrel{(d)}{<}  \, \log \det \Big(  \big( \Ix_{N_t} + \Ex [ \tilde{\Hm}^H \tilde{\Hm} ] \big)  \Ex \big [ (\tilde{\Hm}^H \tilde{\Hm})^{-1} \big ] \Big)  \\
& \stackrel{(e)}{=}  \log \det \big(  (1+N_r) \, \frac{1}{N_r-N_t} \Ix_{N_t}  \big)  \,
\twocolbreak
= \, N_t \log \big ( 1 + \frac{N_t+1}{N_r - N_t} \big ),
\end{align*}
where~$(d),(e)$ follow from Case $1$ and Lemma~\ref{lemma:wishart}, respectively.

\subsection{Case $3$: $N_t=1$ and $\SNR < \frac{1}{\Ex [ ||\tilde{\hv}||^2 ] }$}
\begin{align*}
 \gap =& ~C-R   \\
 & = \Ex \big[ \log{( 1 + \SNR ||\tilde{\hv}||^2 )} \big] + \log \big (\Ex [ \frac{1}{1 + \SNR ||\tilde{\hv}||^2 } ] \big )  \,
\twocolbreak
 \stackrel{(f)}{\leq} \, \log \big ( 1 + \SNR \Ex [||\tilde{\hv}||^2] \big ) +  \log \Big (\Ex \big[ \frac{1}{1 + \SNR ||\tilde{\hv}||^2 } \big] \Big )   \\ 
& \stackrel{(g)}{\leq}  \log e \, \Ex [||\tilde{\hv}||^2] \SNR +   \log e \, \Ex \Big[ \frac{- \SNR ||\tilde{\hv}||^2 }{1 + \SNR ||\tilde{\hv}||^2 } \Big] \,
\twocolbreak
= \, \log e \, \Ex \Big[ \SNR ||\tilde{\hv}||^2 - \frac{\SNR ||\tilde{\hv}||^2 }{1 + \SNR ||\tilde{\hv}||^2} \Big] \\
 & =  \log e \, \Ex \Big[ \frac{||\tilde{\hv}||^4}{1 + \SNR ||\tilde{\hv}||^2} \Big] \, \SNR^2  \,
 < \, 1.45 \, \Ex \big[ ||\tilde{\hv}||^4 \big] \, \SNR^2,
\end{align*}
where~$(f)$ is due to Jensen's inequality and~$(g)$ utilizes $\ln {x} \leq x-1$.


\section{Proof of Corollary~\ref{cor:gap_ptp} }
\label{appendix:gap_ptp}

The results in Case~$1$ and Case~$2$ are straightforward from Corollary~\ref{cor:gap_MIMO}. The proofs are therefore omitted.

\subsection{Case $3$: $\SNR \geq 1$, Nakagami-$m$ fading with $m>1$}
The Nakagami-$m$ distribution with~$m>1$ satisfies the condition~$\Ex \big[ \frac{1}{|\tilde{h}|^2} \big] < \infty $. For a Nakagami-$m$ variable with unit power, i.e.,~$\Ex [|\tilde{h}|^2 ]=1$,  $\Ex \big[ \frac{1}{|\tilde{h}|^2} \big]$ is computed as follows
\begin{align*}
\Ex \big[ \frac{1}{|\tilde{h}|^2} \big]  & = \,   
\frac{2m^m}{\Gamma(m)} \,  \int_{0}^{\infty} {\frac{1}{x^2} x^{2m-1} e^{-mx^2} dx }   \,
\twocolbreak
 = \,   \frac{2m^m}{\Gamma(m)} \, \frac{1}{2 m^{m-1}} \, \int_{0}^{\infty} { y^{m-2} e^{-y} dy }   \\
 = &   \,  \frac{m \, \Gamma(m-1)}{\Gamma(m)}   \, 
 = \,    \frac{m \, \Gamma(m-1)}{(m-1) \, \Gamma(m-1)}   \,
 = \,   1 + \frac{1}{m-1} \, ,
\end{align*}
where~$\Gamma(\cdot)$ denotes the {\em gamma function}. Substituting in~\eqref{gap_NtNr_eq}, $\gap < 1+ \log \big (  1+ \frac{1}{m-1} \big )$.

\subsection{Case $4$: $\SNR \geq 1$, Rayleigh fading}

\begin{lemma} 
For any~$z>0$, the exponential integral function defined by $\bar{E}_1(z)= \int_{z}^{\infty}{ \frac{e^{-t}}{t} dt}$ is upper bounded by
\begin{equation*}
\bar{E}_1(z) < \frac{1}{\log e} \, e^{-z} \, \log(1+\frac{1}{z}).
\vspace{-5mm}
\end{equation*}
\label{lemma:expint}
\end{lemma}

\begin{proof}
See~\cite[Section~5.1]{Handbook_math}.
\end{proof}

Under Rayleigh fading, $|\tilde{h}|^2$ is exponentially distributed with unit power. Hence, 
\begin{align*}
\gap & =  \Ex [ \log{( 1 + \SNR |\tilde{h}|^2 )} ] +  \log \big (\Ex [ \frac{1}{1 + \SNR |\tilde{h}|^2} ] \big )  \,
\twocolbreak
\stackrel{(a)}{\leq} \, \log \big ( 1 + \SNR \Ex [|\tilde{h}|^2] \big ) +  \log \big (\Ex [ \frac{1}{1 + \SNR |\tilde{h}|^2 } ] \big )   \\ 
& \stackrel{(b)}{\leq}  1+ \log \big (  \Ex [ \frac{1}{|\tilde{h}|^2 + \frac{1}{\SNR}} ] \big )  \,
\leq \, 1+ \log \big (\int_{0}^{\infty}{\frac{1}{x+\frac{1}{\SNR}} e^{-x} dx} \big)  \,
\twocolbreak
= \, 1+ \log \big( e^{\frac{1}{\SNR}} \, \int_{\frac{1}{\SNR}}^{\infty}{\frac{1}{y} e^{-y} dy} \big)  \\
& = 1+ \log \big ( e^{\frac{1}{\SNR}} \bar{E}_1(\frac{1}{\SNR}) \big )  \,
\stackrel{(c)}{<} \, 1+ \log \big ( \frac{1}{\log e}  \log ( 1 + \SNR ) \big )  \,
\twocolbreak
< \, 0.48 + \log \big ( \log ( 1 + \SNR ) \big ) \, ,
\end{align*}
where~$(a)$ follows from Jensen's inequality. $(b)$ holds from the condition~$\SNR \geq 1$ and
$(c)$ follows from Lemma~\ref{lemma:expint}.


\section{Proof of Corollary~\ref{cor:gap_MIMO_MAC_K} }
\label{appendix:DoF_MIMO_MAC_K}

\begin{lemma}
For any two independent i.i.d. Gaussian matrices $\Am \in \comp^{r \times m}$,\, $\Bm \in \comp^{r \times q}$ where $r \geq q+1$ whose elements have zero mean and unit variance,
\begin{equation}
\Am^H \big( c \Ix_r + \Bm \Bm^H   \big)^{-1} \Am  \succ \frac{1}{c} \, \boldsymbol{\bar{A}}^H \boldsymbol{\bar{A}},
\label{definiteness}
\end{equation}
where the elements of $\boldsymbol{\bar{A}} \in \comp^{(r-q) \times m}$ are i.i.d. Gaussian with zero-mean and unit variance, 
 and $c$ is a positive constant.
\label{lemma:MIMO_MAC}
\end{lemma}
\begin{proof}
Using the eigenvalue decomposition of $\big( c \Ix_r + \Bm \Bm^H  \big)^{-1}$,
\begin{equation}
\Am^H \big( c \Ix_r + \Bm \Bm^H   \big)^{-1} \Am  = \,  \Am^H \Vm \Dm \Vm^H \Am       \,
= \,  \check{\Am}^H \Dm \check{\Am},
\label{eigen_decomposition}
\end{equation}
where the columns of $\Vm$ are the eigenvectors of $\Bm \Bm^H$. 
The corresponding eigenvalues of $\Bm \Bm^H$ are then in the form $\sigma_1^2,\ldots,\sigma_q^2,0,\ldots,0$. Hence, $q$ of the diagonal entries of~$\Dm$ are in the form $1/(c+\sigma_j^2)$, whereas the remaining $r-q$ entries are $1/c$. Since~$\Vm$ is unitary, then $\check{\Am} \triangleq \Vm^H \Am$ is i.i.d. Gaussian, similar to $\Am$~\cite{random_matrix}. One can rewrite~\eqref{eigen_decomposition} as follows
\begin{align}
\check{\Am}^H \Dm \check{\Am} & = \, \sum_{j=1}^{r-q} \frac{1}{c} \, \check{\av}_j \check{\av}_j^H + \sum_{j=r-q+1}^{r} \frac{1}{c+\sigma_j^2} \, \check{\av}_j \check{\av}_j^H   \,
\twocolbreak
\succ \, \sum_{j=1}^{r-q} \frac{1}{c} \, \check{\av}_j \check{\av}_j^H  \,
= \, \frac{1}{c} \, \bar{\Am}^H \bar{\Am},
\label{eigen_decomposition_2}
\end{align}
where $\check{\av}_j$ is the conjugate transposition of row~$j$ in~$\check{\Am}$, and the columns of the matrix $\bar{\Am}$ are $\check{\av}_j$ for $j \in \{1,\ldots,r-q \}$. The generalized inequality in~\eqref{eigen_decomposition_2} follows since $\boldsymbol{X} + \boldsymbol{Y} \succeq \boldsymbol{X}$ for any two positive semidefinite matrices $\boldsymbol{X},\boldsymbol{Y}$.
\end{proof}

Let $\tilde{\Fm}_{\pi(k)} \triangleq \Ix_{N_r} + \SNR \sum_{l=k+1}^K \tilde{\Hm}_{\pi(l)} \tilde{\Hm}_{\pi(l)}^H$, where~$\pi(\cdot)$ is an arbitrary permutation as described in Section~\ref{sec:MAC_MIMO}. We first bound the sum capacity in~\eqref{capacity_MIMO_MAC} (from above) as follows
\begin{align}
C_{\text{sum}} & \triangleq  \Ex \Big[ \log { \det \big( \Ix_{N_r} + \sum_{k=1}^K \SNR \tilde{\Hm}_k \tilde{\Hm}_k^H \big)} \Big] \nonumber \\
& =  \sum_{k=1}^K \Ex \Big[ \log  \det \big( \Ix_{N_t} + \SNR \tilde{\Hm}_{\pi(k)}^H \tilde{\Fm}_{\pi(k)}^{-1}  \tilde{\Hm}_{\pi(k)} \big) \Big]   \nonumber \\
& \stackrel{(a)}{\leq}  \sum_{k=1}^K \Ex \Big[ \log { \det \big( \Ix_{N_t} + \SNR \tilde{\Hm}_{\pi(k)}^H \tilde{\Hm}_{\pi(k)} \big)} \Big]  \,
\twocolbreak
 \leq \, \sum_{k=1}^K \log { \det \Big( \Ix_{N_t} + \SNR \, \Ex \big[ \tilde{\Hm}_{\pi(k)}^H  \tilde{\Hm}_{\pi(k)} \big] \Big)}   \nonumber \\
& =  K \, \log { \det \big( ( 1 + \SNR N_r ) \Ix_{N_t} \big)}   \,
= \, N_t K \log { ( 1 + \SNR N_r )}    \, 
\twocolbreak
\stackrel{(b)}{\leq} \, N_t K \big( \log {\SNR} +  \log { ( 1 + N_r ) \big)},
\label{gap_MIMO_MAC_1}
\end{align}
where~$(a)$ follows since interference cannot increase capacity, and $(b)$~follows since $\SNR \geq 1$.

Now, we bound (from below) $R_{\text{sum}}$. Since the sum of the rate expressions in both~\eqref{rate_SIMO_MAC_complex} and~\eqref{rate_MIMO_MAC_complex} are equal, we bound each of the $N_t K$ terms in~\eqref{rate_SIMO_MAC_complex}, where the power is allocated uniformly over each virtual user, given by~$\SNR$ as follows
\begin{align}
  & R_{\pi(\ell)} = 
	\twocolbreak
	   - \log \Big( \Ex \big[ \frac{1}{ 1 + \SNR \tilde{\hv}_{\pi(\ell)}^H  (\Ix_{N_r} + \SNR \sum_{j=\ell+1}^L  \tilde{\hv}_{\pi(j)} \tilde{\hv}_{\pi(j)}^H )^{-1} \tilde{\hv}_{\pi(\ell)} } \big] \Big)    \nonumber \\
	& =  - \log \Big( \Ex \big[ \frac{1}{ 1 + \SNR \tilde{\hv}_{\pi(\ell)}^H  (\frac{1}{\SNR} \Ix_{N_r} +  \tilde{\Gm}_{\pi(\ell)} \tilde{\Gm}_{\pi(\ell)}^H )^{-1} \tilde{\hv}_{\pi(\ell)} } \big] \Big)      \nonumber \\
&	\stackrel{(c)}{\geq}  - \log \Big( \Ex \big[ \frac{1}{ 1 + \SNR \check{\hv}_{\pi(\ell)}^H   \check{\hv}_{\pi(\ell)} } \big] \Big)     \,
	> \,  - \log \Big( \Ex \big[ \frac{1}{  \SNR \check{\hv}_{\pi(\ell)}^H   \check{\hv}_{\pi(\ell)} } \big] \Big)     \nonumber \\
	& =   \log \SNR  - \log \Big( \Ex \big[ \frac{1}{  \check{\hv}_{\pi(\ell)}^H   \check{\hv}_{\pi(\ell)} } \big] \Big)     \,
	\twocolbreak
	\stackrel{(d)}{=} \, \log \SNR  +  \log  \big( N_r - (K-\ell+1) \big),
\label{gap_MIMO_MAC_2}
\end{align}
where $\tilde{\Gm}_{\pi(\ell)} \triangleq [\tilde{\hv}_{\pi(\ell+1)} , \ldots  , \tilde{\hv}_{\pi(L)}]$. $(c)$ follows from Lemma~\ref{lemma:MIMO_MAC} where~$\check{\hv}_{\pi(\ell)} \in \comp^{N_r-L+\ell}$ is an i.i.d. Gaussian distributed vector whose elements have unit variance, and $(d)$ follows from Lemma~\ref{lemma:wishart}.

Hence, from~\eqref{gap_MIMO_MAC_1},\eqref{gap_MIMO_MAC_2} the gap is bounded as follows
\begin{align}
\gap \triangleq & C_{\text{sum}} - R_{\text{sum}}   \nonumber  \\
& <  \sum_{\ell=1}^{N_t K} \Big( \log { ( 1 + N_r )} \, - \, \log  \big( N_r - (N_t K-\ell+1) \big) \Big)     \nonumber   \\
& =  \sum_{\ell=1}^{N_t K} \log \big( \frac{ 1 + N_r }{ N_r - (N_t K-\ell+1)} \big)     \,
\twocolbreak
= \, \sum_{\ell=1}^{N_t K} \log \big( \frac{ 1 + N_r }{ N_r - \ell } \big)   \, 
= \, \sum_{\ell=1}^{N_t K} \log \big( 1 + \frac{ \ell +1 }{ N_r - \ell } \big).
\label{gap_MIMO_MAC_3}
\end{align}


\section{Proof of Corollary~\ref{cor:gap_MAC} }
\label{appendix:gap_MAC}

\subsection{Case $1$: $\SNR < \frac{1}{2}$}
\begin{align*}
 \gap =&  \Ex \big[ \log{( 1 + \SNR |\tilde{h}_1|^2 + \SNR |\tilde{h}_2|^2)} \big] 
\twocolbreak
+  \log \Big( \Ex \big[ \frac{ 1 + \SNR |\tilde{h}_1|^2 }{ 1 + \SNR |\tilde{h}_1|^2 + \SNR |\tilde{h}_2|^2 } \big]
\Ex \big[ \frac{1}{ 1 + \SNR |\tilde{h}_1|^2} \big] \Big) \\ 
\leq & \log{ \big( 1 + \SNR \Ex [|\tilde{h}_1|^2] + \SNR \Ex [|\tilde{h}_2|^2] \big )} 
 + \log \Big( \Ex \big[ \frac{1}{ 1 + \SNR |\tilde{h}_1|^2 } \big] \Big) 
\twocolbreak
+  \log \Big( \Ex \big[ \frac{ 1 + \SNR |\tilde{h}_1|^2 }{ 1 + \SNR |\tilde{h}_1|^2 + \SNR |\tilde{h}_2|^2 } \big] \Big ) 
  \\ 
<& \log e \,  \Big ( \SNR \Ex [|\tilde{h}_1|^2] + \SNR \Ex [|\tilde{h}_2|^2] + \Ex \big[ \frac{ - \SNR |\tilde{h}_2|^2 }{ 1 + \SNR |\tilde{h}_1|^2 + \SNR |\tilde{h}_2|^2  } \big] 
\twocolbreak
+ \Ex \big[ \frac{- \SNR |\tilde{h}_1|^2}{ 1 + \SNR |\tilde{h}_1|^2 } \big] \Big ) \\
< & \log e \, \Big ( \Ex \big [ \SNR |\tilde{h}_2|^2 - \frac{ \SNR |\tilde{h}_2|^2 }{ 1 + \SNR |\tilde{h}_1|^2 + \SNR |\tilde{h}_2|^2 } \big ] 
\twocolbreak
+ \Ex \big [ \SNR |\tilde{h}_1|^2 - \frac{ \SNR |\tilde{h}_1|^2 }{ 1 + \SNR |\tilde{h}_1|^2 } \big ] \Big ) \\
= & \log e \, \Big ( \Ex \big [ \frac{ \SNR^2 |\tilde{h}_1|^2 |\tilde{h}_2|^2 + \SNR^2 |\tilde{h}_2|^4}{1 + \SNR |\tilde{h}_1|^2 + \SNR |\tilde{h}_2|^2 } \big ] + 
\Ex \big [ \frac{ \SNR^2 |\tilde{h}_1|^4} { 1 + \SNR |\tilde{h}_1|^2 } \big ] \Big ) \\
\leq & \log e \, \big ( \Ex  [  \SNR^2 |\tilde{h}_1|^2 |\tilde{h}_2|^2 + \SNR^2 |\tilde{h}_2|^4  ] + \Ex  [ \SNR^2 |\tilde{h}_1|^4  ] \big ) \\
= & \log e \, \big ( 1 + \Ex[|\tilde{h}_1|^4] + \Ex [|\tilde{h}_2|^4] \big ) \, \SNR^2  \,
\twocolbreak 
< \, 1.45 \,  \big ( 1 + 2 \Ex[|\tilde{h}_1|^4] \big ) \,  \SNR^2.
\end{align*}

\subsection{Case $2$: $\SNR \geq \frac{1}{2}$ and~$\Ex \big[ \frac{1}{|\tilde{h}|^2} \big] < \infty $}
\begin{align*}
\gap = &  \Ex \big[ \log{( 1 + \SNR |\tilde{h}_1|^2 + \SNR |\tilde{h}_2|^2)} \big] 
\twocolbreak
+ \log \Big( \Ex \big[ \frac{ 1 + \SNR |\tilde{h}_1|^2  }{ 1 + \SNR |\tilde{h}_1|^2 + \SNR |\tilde{h}_2|^2  } \big]
\Ex \big[ \frac{1}{ 1 + \SNR |\tilde{h}_1|^2 } \big] \Big)   \\ 
\stackrel{(a)}{\leq} & \log{ \big( 1 + \SNR \Ex [|\tilde{h}_1|^2] + \SNR \Ex [|\tilde{h}_2|^2] \big )}] 
\twocolbreak
+  \log \Big( \Ex \big[ \frac{ 1 + \SNR |\tilde{h}_1|^2  }{ 1 + \SNR |\tilde{h}_1|^2 + \SNR |\tilde{h}_2|^2 } \big]
\Ex \big[ \frac{1}{ 1 + \SNR |\tilde{h}_1|^2 } \big] \Big)   \\ 
< & \log \Big( ( 1 + 2 \SNR ) \Ex \big[ \frac{1}{ 1 + \SNR |\tilde{h}_1|^2 } \big] \Big)      \\ 
\stackrel{(b)}{\leq} & 2 + \log \big (  \Ex [ \frac{1}{|\tilde{h}_1|^2 + \frac{1}{\SNR}} ] \big )  \, 
< \, 2 + \log \Big (  \Ex \big[ \frac{1}{|\tilde{h}_1|^2} \big] \Big ),
\end{align*}
where~$(a)$ follows from Jensen's inequality and~$(b)$ follows since~$\SNR \geq \frac{1}{2}$.

\subsection{Case $3$: $\SNR \geq \frac{1}{2}$, Nakagami-$m$ fading with $m>1$}
Since the Nakagami-$m$ distribution with~$m>1$ belongs to the class of distributions in Case~$2$, then
\begin{equation*}
\gap < \,  2 + \log \Big (  \Ex \big[ \frac{1}{|\tilde{h}_1|^2} \big] \Big )  \,
\stackrel{(c)}{=} \,  2 + \log \big (  1+ \frac{1}{m-1} \big ),
\end{equation*}
where~$(c)$ follows from the proof of Case~$3$ in Appendix~\ref{appendix:gap_ptp}.

\subsection{Case $4$: $\SNR \geq \frac{1}{2}$, Rayleigh fading}
\begin{equation*}
\gap  \stackrel{(d)}{\leq} \, 2 + \log \big (  \Ex [ \frac{1}{|\tilde{h}_1|^2 + \frac{1}{\SNR}} ] \big )  \, 
\stackrel{(e)}{<} \, 1.48 + \log \big ( \log ( 1 + \SNR ) \big ),
\end{equation*}
where~$(d)$ follows from Case~$2$ and~$(e)$ follows from Case~$4$ in Appendix~\ref{appendix:gap_ptp}.

\bibliographystyle{IEEEtran}
\bibliography{IEEEabrv,References_arxiv}

\end{document}